\documentclass[a4paper,UKenglish,cleveref, autoref, thm-restate]{lipics-v2021}
\nolinenumbers
\hideLIPIcs

\usepackage{amssymb,amsthm,amsmath}
\usepackage{mathtools}
\usepackage{proof-dashed}
\usepackage{stmaryrd}
\usepackage{thmtools}
\usepackage{thm-restate}
\usepackage{tikz}
	\usetikzlibrary{arrows.meta}
	\usetikzlibrary{calc}

\usepackage{FMC}

\newcommand\notdashweak{\makebox[3pt][l]{$\not{\kern-2pt{\vdash}}$}\raisebox{5pt}{$\scriptstyle\weak$}}
\newcommand\spine{\mathsf w}
\newcommand\case[1]{\textbf{#1}}
\newcommand\defn\emph
\newcommand\vc[1]{\vcenter{\hbox{\ensuremath{#1}}}}
\newcommand\msep{\kern1pt{\cdot}\kern1pt}

\newcommand\rp{\mathrel{\smallbin\blacktriangleright}}

\newcommand\ar{\kern1pt{\smallbin\to}\kern1pt}
\newcommand\tim{\kern1pt{\smallbin\times}\kern1pt}

\newcommand\then{\kern1pt{\smallbin\gg}\kern1pt}
\newcommand\bind{\mathbin{
\begin{tikzpicture}[x=1pt,y=-1pt]
\draw[cap=round,join=round,line width=0.7pt] (0,2)--(2,0)--(0,-2) (2,2)--(4,0)--(2,-2) (4.5,1)--(7,1) (4.5,-1)--(7,-1);
\end{tikzpicture}}}

\title{Quantitative Types for the Functional Machine Calculus (Extended Version)}

\titlerunning{Quantitative Types for the FMC}

\author
	{Willem Heijltjes}
	{Department of Computer Science, University of Bath, UK \and \url{http://willem.heijltj.es}}
	{w.b.heijltjes@bath.ac.uk}
	{https://orcid.org/0009-0001-8941-1150}
	{}

\authorrunning{W.\ Heijltjes}

\Copyright{Willem Heijltjes}



\acknowledgements{I am grateful to Vincent van Oostrom for many interesting discussions. Thanks also to the anonymous referees for suggestions and pointers to the literature.}

\begin{document}

\maketitle

\begin{abstract}
The Functional Machine Calculus (FMC, Heijltjes 2022) extends the lambda-calculus with the computational effects of global mutable store, input/output, and probabilistic choice while maintaining confluent reduction and simply-typed strong normalization. Based in a simple call--by--name stack machine in the style of Krivine, the FMC models effects through additional argument stacks, and introduces sequential composition through a continuation stack to encode call--by--value behaviour, where simple types guarantee termination of the machine.

The present paper provides a discipline of quantitative types, also known as non-idempotent intersection types, for the FMC, in two variants. In the weak variant, typeability coincides with termination of the stack machine and with spine normalization, while exactly measuring the transitions in machine evaluation. The strong variant characterizes strong normalization through a notion of perpetual evaluation, while giving an upper bound to the length of reductions. Through the encoding of effects, quantitative typeability coincides with termination for higher-order mutable store, input/output, and probabilistic choice.
\end{abstract}

\section{Introduction}

The Functional Machine Calculus (FMC)~\cite{Barrett-Heijltjes-McCusker-2023,Heijltjes-2022} is a new approach to extending the $\lambda$-calculus with computational effects, preserving confluent reduction and simply-typed strong normalization for the effects of global mutable store, input/output, and probabilistic choice. This paper presents a generalization of the type system to \emph{non-idempotent intersection types}, a notion of \emph{quantitative} types that not only characterizes weak or strong normalisation, but in addition provides quantitative information about the length of reductions.

It is well known that introducing effects into the $\lambda$-calculus renders it non-confluent. The FMC gets around this by starting not from $\beta$-reduction, which indeed is incompatible with global effect operations, but from the (simplified) Krivine machine~\cite{Krivine-2007}. This gives a call--by--name operational semantics to the $\lambda$-calculus in the form of an abstract machine with a single stack, where application pushes its argument and abstraction pops. The FMC then introduces additional stacks to model effects, for example mutable variables as stacks of depth one and a random generator as a probabilistically generated stream (an infinite stack). It further adds sequential composition, naturally interpreted on the stack machine, which allows the encoding of call--by--value evaluation as well as Moggi's computational metalanguage~\cite{Moggi-1991} and Levy's call--by--push--value~\cite{Levy-2003}. Confluent reduction and simple types are then derived from the behaviour of the machine, where types have a natural interpretation as a guarantee of successful termination of the machine~\cite{Heijltjes-2022}.

The FMC continues to be developed, for example with an extension to \emph{exception handling}~\cite{Heijltjes-2025} and as a symmetric variant for relational programming~\cite{Barrett-Castle-Heijltjes-2024}, and there are many remaining questions regarding its type system and possible extensions to it. The simply-typed FMC is constrained in its expressiveness, for instance requiring the consumption of a fixed number of probabilistic values from the stream representing a random generator. Like monadic encodings of effects such as store, semantically it remains a Cartesian closed category~\cite{Barrett-Heijltjes-McCusker-2023}. Richer type systems are required (and expected) to more accurately capture the nature of effects, as demonstrated by a recent type system for probabilistic termination~\cite{Heijltjes-Majury-2025}. The aim of the present paper is to take a next step in this direction, moving from simple types to intersection types.

Intersection type systems~\cite{Coppo-Dezani-1980,Pottinger-1980}, which characterize exactly the normalizing terms of a calculus, have become an important semantic tool (see~\cite{Bono-Dezani-2020} for a recent overview). They are a form of finite polymorphism, where a term may be assigned multiple types. Conversely, a collection of multiple types is inhabited by the intersection of the inhabitants of its constituents---hence the name \emph{intersection types}. The original \emph{idempotent} intersection types, where type collections are sets, were later joined by \emph{non-idempotent} variants~\cite{Bucciarelli-Kesner-Ventura-2017,Kfoury-Wells-1999,Neergaard-Mairson-2004}, where collections are multisets. In addition to characterizing normalization, these resource-sensitive versions provide \emph{quantitative} information about terms such as the maximal length of reductions. They are closely related to a range of other models and techniques, including resource calculi and differential $\lambda$-calculus~\cite{Boudol-1993,Ehrhard-Regnier-2003,Kfoury-2000}, relational models~\cite{Laird-Manzonetto-McCusker-Pagani-2013,Paolini-Piccolo-Ronchi-2017}, and game semantics~\cite{Ong-2017}. The relation between both variants manifests at the semantic level as a connection between Scott models and relational models~\cite{Ehrhard-2012-CSL,Ehrhard-2012-TCS}. Intersection types have been studied from various other angles including category theory~\cite{Mazza-Pellissier-Vial-2018} and logic and proof theory~\cite{Ehrhard-2020,Guerrieri-Heijltjes-Paulus-2021,Pimentel-Ronchi-Roversi-2012}.

Importantly for the present purpose, intersection types have been used to study computational effects~\cite{Blaauwbroek-2017,Davies-Pfenning-2000,Dezani-Giannini-Ronchi-2009}. In particular, intersection type systems exist for probabilistic choice~\cite{Antonelli-DalLago-Pistone-2022,Breuvart-DalLago-2018}, mutable store~\cite{Alves-Kesner-Ramos-2023,deLiguoro-Treglia-2021,deLiguoro-Treglia-2023}, and monadic effects in general~\cite{Gavazzo-Treglia-Vanoni-2024}.

Where these efforts start from untyped effects and introduce tailored intersection types to capture their semantics, the development in this paper is in the opposite direction: it starts from a simply-typed calculus for effects, the FMC, and derives intersection types essentially by the standard generalization, which is to replace the input types of a function with a collection type. The litmus test is whether both approaches arrive at the same solution. Validating both views, this is confirmed in Section~\ref{sec:store}, where it is demonstrated that modulo the distinction between idempotent and non-idempotent types, the type system presented here agrees with that of De'Liguoro and Treglia for higher-order store~\cite{deLiguoro-Treglia-2021}. Appendix~\ref{sec:QGM} extends the comparison to the non-idempotent type system for store of Alves, Kesner, and Ramos~\cite{Alves-Kesner-Ramos-2023}.

The contributions of the paper are as follows. The overall result is to give a system of quantitative types for the FMC that is implicitly reflected on encoded effects and strategies: higher-order store, input/output, and probabilistic choice, the call--by--name and call--by--value $\lambda$-calculus~\cite{Plotkin-1975}, the computational metalanguage~\cite{Moggi-1991}, and call--by--push--value~\cite{Levy-2003} (see Section~\ref{sec:store} for some of their encodings, and~\cite{Heijltjes-2022} for the remainder). The type system has two variants, \emph{weak} and \emph{strong}. The main theorem for the weak system, Theorem~\ref{thm:weak}, is that weak quantitative types characterize both termination of the machine, where types give the exact number of steps in machine evaluation, and spine normalization. That machine termination and spine normalization coincide then holds untyped, which is a new result for the FMC relating the machine to reduction. The main result for strong quantitative types, Theorem~\ref{thm:strong}, is that typing coincides with strong normalization and gives an upper bound to reduction length. The development is via a new notion of \emph{perpetual evaluation} for the FMC, an inductive definition of iterated weak head reduction to normal form analogous to \emph{perpetual strategies} for the $\lambda$-calculus~\cite{Barendregt-1984,Bergstra-Klop-1982,Raamsdonk-Severi-Sorensen-Xi-1999}. 

An interesting observation is that weak typing corresponds to machine termination, an important and natural notion for the FMC, but not to weak normalization. That is because, unlike for the call--by--name $\lambda$-calculus, the machine may terminate with a non-empty stack of return values, about which the type system gives no guarantees (i.e.\ they may be non-terminating).

The paper uses a variant definition of the FMC that differs superficially from the original~\cite{Heijltjes-2022}. The changes and their benefits are explained after the calculus is introduced in the next section. Confluence holds~\cite{Heijltjes-2025}, but since it is unpublished at time of writing, the present exposition is formulated so as not to rely on it.

The development aims to put the emphasis on careful definitions, so that proofs become straightforward case analyses and structural inductions. This manifests particularly in the perpetual evaluation relation used to relate strong types and strong normalization.  In this extended version of the conference paper, proofs that are indeed of this kind have been relegated to the Appendix.

\section{The Functional Machine Calculus}
\label{sec:FMC}

This section will define the FMC: its syntax, the abstract machine, a quantitative big-step semantics, its reduction relations, and the type system. 

The primary intuition for the FMC is as an instruction language for a machine with multiple stacks, indexed by a set of \emph{locations} $\{a,b,c,\dots\}$, and together called a \emph{memory}. The calculus interacts with individual stacks through \emph{push} and \emph{pop} constructs: $\term{[N]a.M}$ pushes the term $\term N$ to the stack $a$ and continues as $\term M$, and $\term{a<x>.M}$ pops from the stack $a$, using the popped term as $\term x$ in $\term M$. This generalizes the evaluation of $\lambda$-terms on the (simplified) Krivine machine~\cite{Krivine-2007}, where application pushes its argument, and abstraction pops. The standard, call--by--name $\lambda$-calculus~\cite{Plotkin-1975} then embeds by choosing a \emph{default} location $\lambda$, by letting $M\,N=\term{[N]l.M}$ and $\lambda x.M=\term{l<x>.M}$. For brevity, the location $\lambda$ may be omitted, writing instead $\term{[N].M}$ respectively $\term{<x>.M}$. The locations of the FMC may encode \emph{reader/writer} effects: higher order store as stacks of depth one, I/O as separate input and output streams, and non-determinism and probabilistic choice as non-deterministically and probabilistically generated input streams (see~\cite{Heijltjes-2022} for details).

The FMC further includes \emph{sequencing} $\term{M;N}$ and its unit \emph{skip} $\term*$. These are composition and identity of FMC terms viewed as \emph{memory transformers}: computations with an input memory and output memory. The big-step semantics, below, uses this interpretation directly, while it is implemented on the machine  via a continuation stack, in standard fashion. 

\begin{definition}[Syntax]
The \emph{terms} of the Functional Machine Calculus are given by the following grammar.
\[
	\term{M,N,P}~\Coloneqq~
		\term{x}
~\mid~	\term{a<x>.M}
~\mid~	\term{[N]a.M}
~\mid~	\term{*}
~\mid~	\term{M;N}	
\]
From left to right these are: a \emph{variable} $\term x$, an \emph{abstraction} or \emph{pop} $\term{a<x>.M}$ on location $a$ which binds $x$ in $\term M$, an \emph{application} or \emph{push} $\term{[N]a.M}$ of term $\term N$ on location $a$, the \emph{unit} or \emph{skip} $\term*$, and \emph{sequencing} $\term{M;N}$.
\end{definition}

The operational meaning of the five constructs will be expressed through a small-step semantics, in the form of a simple stack machine, and a big-step semantics, in the form of an inductive evaluation relation. The abstract machine will use the following components. An \emph{operand stack} or \emph{argument stack} $S$ is a stack of terms, represented with the head on the right and using $\e$ for the empty stack. A \emph{memory} is a family of operand stacks indexed in a set of locations $A$. A \emph{continuation stack} is a stack of terms with the head to the left.
\[
	S,T~\Coloneqq~\e~\mid~S\,\term M
\qquad
	S_A~\Coloneqq~\{S_a\mid a\in A\}
\qquad
	K,L~\Coloneqq~\term M\,K~\mid~\e
\]
For a memory $S_A$, the stack $a$ being empty, $\e_a$, or outside the family, $a\notin A$, are equivalent. Thus we may fix a set of locations $A$ going forward, noting that any given term uses only a fixed, finite set of locations, determined syntactically. Let $S_A\msep\term M_a$ denote the memory $S_A$ with $\term M$ pushed to the stack indexed by $a$.

\begin{definition}[Abstract machine]
A \emph{machine state} $(S_A,\term M,K)$ is a triple of a memory $S_A$, a term $\term M$, and a continuation stack $K$. The \emph{transitions} or \emph{steps} of the machine given by the following rules.
\[
\begin{array}{cc}
	\step {S_A}{[N]a.M}K  {S_A\msep\term N_a} M K		
&   \step {S_A}{M;N}K  {S_A}M{\term N\,K}
\\ \\
	\step {S_A\msep\term N_a}{a<x>.M}K {S_A}{\{N/x\}M}K	
&	\step {S_A}*{\term M\,K} {S_A}MK
\end{array}
\]
A \emph{final state} is of the form $(S_A,\term *,\e)$ and a \emph{failure state} of either of the forms $(S_A,\term x,K)$ or $(S_A,\term{a<x>.M},K)$ where $S_a=\e$. A \emph{run} is a sequence of steps, written as below, where the \emph{length} of the run $n\geq 1$ is the number of states it traverses (which may be omitted from the notation). A run is \emph{successful} if it terminates in a final state.
\[
\steps 
  {S_A} M K 
  {T_A} N L
\,\raisebox{2pt}{$\scriptstyle n$}
\]
\end{definition}

The following big-step semantics describes the successful runs on the machine, for a given term, as an inductive \emph{evaluation} relation $\evalarrow_n$ between memories. The length of a run will correspond exactly to the number of rule instances in a derivation for $\evalarrow_n$, recorded as the number $n$. This correspondence, and a similar one with the type system defined later in this section, is the reason to measure the length of a run by its \emph{states} and not its \emph{steps}: a unit term $\term{*}$ is introduced by a rule, so correspondingly the unit run will have length one, not zero.

\begin{definition}[Quantified big-step semantics]
The quantified \emph{evaluation} relation $\evalarrow_n$ between an input memory, a term, and an output memory is defined inductively as follows.
\[
\evalRule {} {\eval[1]{S_A}*{S_A}}
\qquad
\evalRule {\eval[m]{R_A}M{S_A} \qquad \eval[n]{S_A}N{T_A}} {\eval[m+n+1]{R_A}{M;N}{T_A}}
\]
\[
\evalAlign 
  [n]   {S_A\msep\term N_a}       M  {T_A}
  [n+1] {S_A}               {[N]a.M} {T_A}
\qquad  
\evalAlign 
  [n]   {S_A}             {\{N/x\}M} {T_A}
  [n+1] {S_A\msep\term N_a} {a<x>.M} {T_A}
\]
\end{definition}

Both the machine and the evaluation relation are \emph{deterministic}: there is exactly one rule for each constructor (except none for the variable). Evaluation is thus a partial function, undefined where the machine is non-terminating or ends in a failure state. Note that it is not inductive on terms due to the substitution in the rule for abstraction, which handles the \emph{variable} construct.

\begin{restatable}[Small-step and big-step semantics agree]{proposition}{smallStepBigStep}
\label{prop:semantics agree}
Successful machine runs of length $n$ correspond to quantified evaluation $\evalarrow_n$:
\[
\steps 
  {S_A} M \e 
  {T_A} * \e
\,\raisebox{2pt}{$\scriptstyle n$}
\quad\iff\quad 
	\eval[n]{S_A}M{T_A}
\]
\end{restatable}

We will see two reduction relations. Regular reduction $\rw$ is a rewrite relation, closed under all contexts, and spine reduction $\rw_\spine$ is restricted to spine contexts, which means everywhere except within the argument $\term{N}$ of an application $\term{[N]a.M}$. 

\begin{definition}[Reduction]
The \emph{reduction rules} are as follows.
\[
\begin{array}{lrcll}
		\text{Beta:}          & \term{[N]a.a<x>.M} & \rw & \term{\{N/x\}M}
\\[2pt]	\text{Passage:}       & \term{[N]b.a<x>.M} & \rw & \term{a<x>.[N]b.M} & (a\neq b, x\notin\fv N)
\\[2pt] \text{Next:}          & \term{*;M}		   & \rw & \term{M}
\\[2pt] \text{Prefix (pop):}  & \term{a<x>.N ; M}  & \rw & \term{a<x>.(N;M)}  & (x\notin\fv M)
\\[2pt] \text{Prefix (push):} & \term{[P]a.N ; M}  & \rw & \term{[P]a.(N;M)}
\\[2pt] \text{Associate:}     & \term{(P;N) ; M}   & \rw & \term{P ; (N ; M)}
\end{array}
\]
\emph{Reduction contexts} $\term C$ are given by the following grammar.
\[
	\term C~\Coloneqq~
		\term{\{\}}
~\mid~	\term{a<x>.C}
~\mid~	\term{[M]a.C}
~\mid~	\term{[C]a.M}
~\mid~	\term{C;M}
~\mid~	\term{M;C}
\]
\emph{Spine contexts} $\term W$ are given by the following grammar.
\[
	\term{W}
~\coloneqq~ \term{\{\}}
~\mid~      \term{a<x>.W}
~\mid~      \term{[M]a.W}
~\mid~      \term{W;M}
~\mid~      \term{M;W}
\]
The (capturing) substitution of the hole $\term{\{\}}$ in a context $\term C$ by a term $\term M$ is written $\term{C\{M\}}$. The \emph{reduction} relation $\rw$ is given by closing the reduction rules $\term M\rw\term N$ under any context, $\term{C\{M\}}\rw\term{C\{N\}}$. \emph{Spine} reduction $\rw_\spine$ is given by closing under spine contexts, $\term{W\{M\}}\rw_\spine\term{W\{N\}}$. The reflexive-transitive closure of a relation $\rw$ is written $\rws$, and reduction to normal form as $\rwn$.
\end{definition}

Without the beta-rule, reduction is strongly normalizing.

\begin{proposition}
\label{prop:non-beta SN}
Non-beta reduction is SN.
\end{proposition}

\begin{proof}
The following pair of measures strictly decreases. Writing $|\term{M}|$ for the size of a term $\term M$, first the sum over $|\term M|$ for every subterm $\term{M;N}$ (equal for \emph{passage}, decreasing for other rules); second the sum over $|\term M|$ for every subterm $\term{[N]a.M}$ (decreasing for \emph{passage}).
\end{proof}

The above definitions differ from the original presentation~\cite{Heijltjes-2022} in the following ways, and for the following reasons. First, instead of introducing sequencing $\term{M;N}$ as a primitive, the original syntax used a sequential variable $\term{x.M}$ and defined sequencing as an operation. The two formulations are equivalent modulo the reductions for sequencing (the Next, Prefix, and Associate reductions), whose normal forms use sequencing exclusively as $\term{x;M}$. The original formulation reduced the number of reduction rules and avoided the need for a continuation stack on the machine, emphasizing the view of terms as sequences of machine instructions. Present thinking however prefers simplicity not through fewer but through more elementary components, viewing the more refined reduction relation and the continuation stack as benefits rather than obstacles. The current syntax emphasizes the clean combination of $\lambda$-calculus and sequencing and more readily accommodates current and future extensions such as exceptions~\cite{Heijltjes-2025} and concurrency.

Second, the original formulation used a single yet complex reduction rule, below, where the ellipsis represents any number of abstractions and applications on locations other than $a$. 
\[
	\term{[N]a\dots a<x>.M}\rw\term{\dots\{N/x\}M}
\]
Again the emphasis has shifted to multiple simpler rules, Beta and Passage. These follow the natural case analysis for an application meeting an abstraction, whether they have the same or different locations, and give more structured normal forms, where all abstractions precede all applications, as we will see in Sections~\ref{sec:spine} and~\ref{sec:perp}.

The main technical results of the paper, in the following sections, will establish the relation between the operational semantics, reduction, and quantitative types. The latter are defined next.

\begin{definition}[Quantitative types]
\defn{Quantitative types} for the FMC are built up in four mutually recursive layers : \defn{computation types}, \defn{collection types}, \defn{vector types}, and \emph{memory types}, given by the following grammars.
\[
\begin{array}{lr@{~~}c@{~~}l}
	    \text{Computation types:}	& \type{t}         & \Coloneqq & \type{!!i => !!k}
\\[2pt]	\text{Collection types:}	& \type{i,k,m,n,p} & \Coloneqq & \type{[t_1,\dots,t_n]}
\\[2pt]	\text{Vector types:}		& \type{!i}        & \Coloneqq & \type{i_1\dots i_n}
\\[2pt]	\text{Memory types:}		& \type{!!i}       & \Coloneqq & \{\type{!i}_a\mid a\in A\}	
\end{array}
\]
Collection types are multisets, i.e.\ considered modulo associativity, symmetry, and unitality, where multiset union is denoted as a sum, $\type{i+k}$.
\end{definition}

These four categories reflect the components of the operational semantics. Terms will be typed with computation types (or simply \emph{types}) $\term{M:t}$, as will continuation stacks $K:\type t$. Collection types may be viewed as \emph{value types}, along the computation/value distinction of call--by--push--value~\cite{Levy-2003}, and apply to the argument term in an application and to the terms on an argument stack (note the continuation stack holds terms as \emph{computations}). Argument stacks themselves are typed by vector types, $S:\type{!i}$, and memories by memory types, $S_A:\type{!!i}$.

Note that the base cases for the definition are given by empty collection types $\type{[]}$ and empty vector and memory types, both denoted $\type{\e}$. Composition of type vectors is by juxtaposition, $\type{!i\,!k}$, extended to memories point-wise, $\type{!!i\,!!k}=\{\type{!i_a\,!k_a}\mid a\in A\}$. The singleton family $\type{a(!i)}$ has $\type{!i}$ at location $a$ and $\type\e$ at any other location. The default location $\lambda$ may again be omitted, writing $\type{!i}$ for $\type{\lambda(!i)}$.

Since stacks are first-in-first-out, identity terms are of the form below, with pops and pushes in inverse order to each other. To match this in the type system, type vectors on the left of an implication are implicitly reversed, so that identity types may be written $\type{!!i => !!i}$.
\[
	\term{a_n<x_n>..a_1<x_1>.[x_1]a_1..[x_n]a_n.*~:~a_n(i_n)..a_1(i_1) => a_1(i_1)..a_n(i_n)}
\] 
A \emph{typing context} $\term{G,D,E,F}$ is a function from variables to collection types, written as a sequence $\term{G = x_1:i_1,,x_n:i_n}$ where any other variable $\term y$ is implicitly assigned the empty type, $\trm{G("1y"0)}=\type{[]}$. The \emph{union} or \emph{sum} of two contexts $\term{G+D}$ is given pointwise: $\trm{(G+D)("1x"0)}=\trm{G("1x"0)}+\trm{D("1x"0)}$.

\begin{definition}[The quantitatively typed FMC]
The typing rules in Figure~\ref{fig:types} define the \defn{quantitatively-typed Functional Machine Calculus}, in two variants: \emph{weak} and \emph{strong}, with typing judgements $\term{G |w M:t}$ respectively $\term{G |s M:t}$ that assign the term $\term M$ the type $\type t$ in the context $\Gamma$. The weak type system is given by the rules of Figure~\ref{fig:types} with unlabelled $(\term{|-})$ and weak judgements $(\term{|w})$, and the strong system  by those with unlabelled $(\term{|-})$ and strong judgements $(\term{|s})$. The \emph{weight} $n$ of a derivation is the number of rule instances for abstraction, application, sequencing, and unit, and may be indicated in a judgement as $\term{|_n}$.
\end{definition}

\begin{figure}
\[
\begin{array}{l@{\qquad}c}
		\text{Variable:}			&\vc{\infer[\rr x]{\term{x:[t] |- x:t}}{}}
\\ \\	\text{Abstraction:} 		&\vc{\infer[\rr l]{\term{G |- a<x>.M : a(i)\,!!k => !!m}}{\term{G, x:i |- M: !!k => !!m}}}
\\ \\	\text{Application (weak):} 	&\vc{\infer[\rr a]{\term{G+D |w [N]a.M : !!i => !!m}}{\term{G |w N : i} && \term{D |w M : a(i)\,!!k => !!m}}}
\\ \\	\text{Unit:}				&\vc{\infer[\rr *]{\term{|- * : !!i => !!i}}{}}
\\ \\	\text{Sequencing:}			&\vc{\infer[\rr ;]{\term{G+D |- N;M : !!i => !!m}}{\term{G |- N : !!i => !!k} && \term{D |- M : !!k => !!m}}}
\\ \\	\text{Collection:}			&\vc{\infer[\rr c]{\term{G_1 ++ G_n |- M : [t_1,\dots,t_n]}}{(\term{G_i |- M : t_i})_{1\leq i\leq n}}}
\\ \\   \hline
\\   	\text{Memory:}              &\vc{\infer[\rr e]{\term{|- ?{\e_A} : \e}}{} \qquad \infer[\rr o]{\term{|- ?{S_A\msep\term M_a} : !!i\,a(k)}}{\term{|- ?{S_A}: !!i} && \term{|- M:k}}}
\\ \\	\text{Continuation stack:}	&\vc{\infer[\rr i]{\term{|- ?{\e} : !!i => !!i}}{} \qquad \infer[\rr k]{\term{|- M\,K : !!i => !!m}}{\term{G |- M : !!i => !!k} && \term{|- K : !!k => !!m}}}
\\ \\	\text{State:}				&\vc{\infer[\rr ?]{\term{|- ??{S_A}MK : \e => !!m}}{\term{|- ?{S_A} : !!i} && \term{|- M : !!i => !!k } && \term{|- K : !!k => !!m}}}
\\ \\	\hline 
\\		\text{Application (strong):}&\vc{\infer[\rr A]{\term{G+D+E |s [N]a.M : !!k => !!m}}{\term{G |s N : i} && \term{D |s M : a(i)\,!!k => !!m} &&\term{E |s N:t }}}
\\ \\	\text{Weakening:}			&\vc{\infer[\rr W]{\term{G+D |s M : t}}{\term{G |s M : t}}}
\end{array}
\]
\caption{The quantitatively typed Functional Machine Calculus}
\label{fig:types}
\end{figure}

The first section of Figure~\ref{fig:types} gives the weak system for \emph{terms}, the next section extends the weak system to memories, continuation stacks, and machine states, and the final section gives the specific rules for the strong type system (which is defined only for terms, not states). When it is clear from the context whether a judgement concerns the weak or the strong type system, the subscripts on $\term{|w}$ and $\term{|s}$ may be dropped.

The key difference between both systems is in the third premise to the strong application rule. In both application rules, the first premise assigns a type to the argument $\term N$ which may be empty in the case where it will not be used by $\term M$. This corresponds to machine evaluation, where $\term N$ may be discarded without evaluation, but not to strong normalization, where $\term N$ must itself be normalizing. The strong application rule therefore requires $\term N$ to be typed separately, a standard approach~\cite{Bono-Dezani-2020,Bucciarelli-Kesner-Ventura-2017}, here via the third premise.

Typing judgements for machine states are of the form $\term{|- ??{S_A}MK : \e=>!!i}$, with empty context and input type. This prevents failure states from being typed: the empty context means there are no computation types for free variables, and an empty input type means there must always be sufficient items in memory for any abstraction to pop. \emph{Progress}, the idea that a machine state is either final or has a transition, is thus immediate from types, once it is proved that machine transitions preserve typing in Lemma~\ref{lem:steps} below.

Non-idempotent intersection types for the $\lambda$-calculus~\cite{Bucciarelli-Kesner-Ventura-2017} embed into quantitative FMC types in the manner of simple types~\cite{Heijltjes-2022}, by the following interpretation, using the default location $\lambda$ (omitted from the notation). The weak type system then embeds the system $\mathcal W$ for $\lambda$-terms of~\cite{Bucciarelli-Kesner-Ventura-2017}, while the strong system differs from system $\mathcal S$ in two minor ways: in the presence of the weakening rule, and for the strong application rule, in requiring a witness also when the argument has a non-empty collection type.
\[
\begin{array}{lr@{}l}
	\text{Base type:}      &                                  o &{}= \type{\e=>\e}
\\	\text{Arrow types:}    & [\sigma_1,\dots,\sigma_n] \to \tau &{}= \type{[\sigma_1,\dots,\sigma_n]\,!k => !m} \quad\text{where}~\tau=\type{!k => !m}
\end{array}
\]

\section{Machine Termination}
\label{sec:machine termination}

This section will show that the successful evaluation and weak typeability coincide, where the weights agree. First, in the weak type system, the following \emph{substitution} typing rule is admissible: it can be introduced without altering the expressiveness of the type system.
\[
	\infer-[\rr s]{\term{G+D |- \{N/x\}M : t}}{\term{G |- N:i} && \term{D, x:i |- M:t}}
\]

\begin{restatable}[Substitution, weak]{lemma}{substitutionWeak}
\label{lem:substitutionWeak}
There are derivations $\term{G |w_n N:n}$ and $\term{D , x:n |w_m M:t}$ if and only if there is a derivation $\term{G+D |w_k \{N/x\}M :t}$ where $k=n+m$.
\end{restatable}

For the weakly typed calculus, a machine step from a typed state leads to a typed state of weight exactly one less.

\begin{restatable}[Evaluation reduces typing weights]{lemma}{evaluationReducesWeights}
\label{lem:steps}
If 
\[
	\step {S_A}MK {T_A}NL
\]
then $\term{|w_n ??{S_A}MK : \e=>!!i}$ if and only if $\term{|w_{n-1} ??{T_A}NL : \e=>!!i}$.
\end{restatable}

It follows that weak types guarantee termination of the machine while measuring the length of a run.

\begin{theorem}[Weak types quantify evaluation]
\label{thm:types => termination}
For a typed state $\term{|w_n ??{S_A}M\e : \e => !!i}$ there exists a typed memory $\term{|w_m ?{T_A} : !!i}$ such that $\eval[k]{S_A}M{T_A}$ where $n = m+k$.
\end{theorem}

\begin{proof}
Recall that failure states cannot be typed. A typed state $\term{|w_n ??{S_A}MK : \e=>!!i}$ is then either final or has a machine transition. In the case of a transition, by Lemma~\ref{lem:steps} the new state is typed with weight one less. Thus, the machine continues taking steps, preserving types but reducing weights, until it reaches a final state, giving a run of length $k$ as below left, so that $\eval[k]{S_A}M{T_A}$ by Proposition~\ref{prop:semantics agree}. The final state $\term{??{T_A}*\e}$ is typed as below right, with weight $m+1$ where $\term{|w_m ?{T_A} : !!i}$. This represents a run of length $1$, so that a run of length $k$ starts from a weight $m+k=n$.
\[
\steps {S_A} M \e {T_A} * \e\,\raisebox{2pt}{$\scriptstyle{k}$}
\qquad\qquad
\vc{\infer[\rr ?]{\term{|w_{m+1} ??{T_A}*\e : \e => !!i}}{\term{|w_m ?{T_A} : !!i} && \infer{\term{|w_1 * : !!i => !!i }}{} && \infer{\term{|w_0 \e : !!i => !!i}}{}}}
\qedhere
\]
\end{proof}

For the reverse, that all terminating states can be typed, it is now sufficient to give a type to final states $(T_A,\term *,\e)$. We may do so simply by typing each term in the memory $T_A$ with the empty collection type $\type{[]}$. Denote by $\type{[]^n}$ the type vector of $\type{[]..[]}$ of length $n$, and by $\type{[]^f}$ for a function $f:A\to\mathbb N$ the memory type $\{\type{[]^{f(a)}}\mid a\in A\}$. Let $|T|$ denote the length of a stack $T$ and let $|T_A|=f$ where $f(a)=|T_a|$ denote the dimensions of a memory $T_A$. For types, likewise let $|\type{!i}|$ denote the length of a type vector and $|\type{!!i}|=f$ where $f(a)=|\type{!i_a}|$ the dimensions of a memory type.

\begin{theorem}[Evaluation implies weak typing]
\label{thm:termination => types}
If $\eval[n]{S_A}M{T_A}$ then $\term{|w_n ??{S_A}M\e : \e => []^f}$ where $f=|T_A|$.
\end{theorem}

\begin{proof}
The memory $T_A$ may be typed $\term{|w_0 ?{T_A}:[]^f}$, and its final state $\term{|w_0 ??{T_A}*\e: \e =>[]^f}$. By Proposition~\ref{prop:semantics agree} there is a run of length $n$ to this state from $(S_A,\term M,\e)$. By Lemma~\ref{lem:steps} each step in this run, backwards from the final state, gives a previous state with the same type but weight one greater, so that by induction $\term{|w_n ??{S_A}M\e : \e => []^f}$.
\end{proof}

Combining both theorems, since a typeable state $\term{|w ??{S_A}M\e: \e=>!!i}$ terminates with a memory $T_A:\type{!!i}$, which is also typed by $\type{[]^f}$ where $f=|\type{!!i}|$, we have $\term{|w_n ??{S_A}M\e: \e =>[]^f}$ where the weight $n$ of the latter is exactly the number of machine steps to successful termination.

\section{Spine Normalization}
\label{sec:spine}

The present section will establish that weak typing and successful evaluation further coincide with spine normalization. First we observe the shape of spine normal forms and demonstrate that they may be typed.

\begin{restatable}[Spine normal forms]{proposition}{spineNF}
\label{prop:spine NF}
The normal forms of spine reduction $\rw_\spine$ are the terms $\term{\WW}$ given by the following mutually recursive grammars.
\[
\begin{array}{l@{~}l@{~}l@{~}l@{~}l}
	\term{\WW} &\Coloneqq& \term{a<x>.\WW} &\mid& \term{\VV}
\\	\term{\VV} &\Coloneqq& \term{[M]a.\VV} &\mid& \term{x;\WW} ~\mid~ \term{x}  ~\mid~ \term*
\end{array}
\]
\end{restatable}

\begin{restatable}[Spine normal forms are weakly typeable]{lemma}{spineNFTyped}
\label{lem:spine NF => typed}
For any $\term{\WW}$ there is a derivation $\term{G |w \WW : !!i => []^f}$.
\end{restatable}

Spine reduction preserves and reflects typing, where weights are strictly reduced in the reduction rules \emph{beta} $(\term{[N]a.a<x>.M}\rw\term{\{N/x\}M})$ and \emph{next} $(\term{*;M}\rw\term M)$, and preserved otherwise.

\begin{restatable}[Weighted subject reduction and expansion]{lemma}{weightedSpine}
\label{lem:spine reduction}
If $\term M\rw_\spine\term N$ then $\term{G |w_m M:t}$ if and only if $\term{G |w_n N:t}$, where $m=n+2$ in the case of a beta-step or next-step and $m=n$ otherwise.
\end{restatable}

It now follows that weak typeability and spine normalization coincide.

\begin{theorem}[Weak typing characterizes spine normalization]
\label{thm:types <=> spine}
A term is weakly typeable if and only if it is spine-normalizing.
\end{theorem}

\begin{proof}
From left to right, since non-beta steps are strongly normalizing by Proposition~\ref{prop:non-beta SN}, an infinite reduction must contain infinitely many beta-steps. By Lemma~\ref{lem:spine reduction} spine beta-reduction decreases the weight of the typing derivation, while other steps do not increase it. Then spine reduction is terminating. From right to left, proceed by induction on the sum length of all spine reduction paths. Spine-normal terms are typeable by Lemma~\ref{lem:spine NF => typed}, and for a step $\term M\rw_\spine\term N$ if $\term N$ is typeable then so is $\term M$ by Lemma~\ref{lem:spine reduction}.
\end{proof}

Since weak typing characterizes both spine normalization and machine termination (by Theorems~\ref{thm:types => termination} and~\ref{thm:termination => types}), the natural corollary would be that machine termination and spine normalization coincide. In one direction, this follows directly: by Theorem~\ref{thm:termination => types} termination implies typeability, and by Theorem~\ref{thm:types <=> spine} typeability implies spine normalization.

In the other direction, however, there is still a small gap to be closed: a spine-normalizing term $\term M$ may be typeable as $\term{G |w M: !!i => !!k}$, but the machine is only guaranteed to run for states of type $\type{\e => !!k}$. In other words, for $\term M$ to run successfully, a memory $S_A$ of type $\term{|w ?{S_A}:!!i}$ must be provided, and each variable $\term{x:i}$ in $\term G$ must be substituted by a term $\term N$ typed $\term{|w N:i}$. However, these are not guaranteed to exist, since not all quantitative types are inhabited. This is in contrast with simple types for the FMC, which \emph{are} inhabited~\cite[Remark 3.7]{Heijltjes-2022}.

\begin{proposition}[No type inhabitation]
There exist types $\type{t}$ such that $\notdashweak\term{M:t}$ for any $\term M$.
\end{proposition}

\begin{proof}
Since types are preserved under spine reduction, if a type is inhabited by a term, it is inhabited by its spine normal form. The only normal form inhabiting the type $\type{\e=>\e}$ is $\term{|- * : \e=>\e}$ (it cannot be an abstraction, since the type has no inputs; it cannot contain free variables, since the context is empty; and it cannot be an application, since the type has no outputs). Similarly, the only spine normal forms inhabiting $\type{\e => [\e=>\e]}$ are of the form $\term{|- [M].* : \e => [\e =>\e ]}$. It follows that the collection type $\type{[ \e=>\e , \e => [\e=>\e] ]}$ is not inhabited, and so neither is the type $\type{\e => [ \e=>\e , \e => [\e=>\e] ]}$.
\end{proof}

To connect spine normalization and termination, for every spine normal term $\term\WW$ a memory $S_A$ and substitution map $\sigma$ will be given such that for $S_A$ and $\term{\sigma\WW}$ the machine terminates. A \emph{substitution map} $\sigma=\{\term{M_1}/\term{x_1},\dots,\term{M_n}/\term{x_n}\}$ is a finite function from variables to terms, applied as $\term{\sigma M}$ to a term $\term M$ as a simultaneous substitution. To allow a successful machine run for $\term{\sigma\WW}$, the memory $S_A$ and substitution map $\sigma$ must be such that, at any point in the computation, there are sufficient terms on all stacks to pop from. To this end the following notions are defined: for a natural number $d$, a \emph{memory of dimension $d$} will have at least $d$ terms of dimension $d-1$ on each stack, and a \emph{term of dimension $d$} consists of sufficient applications to generate a memory of dimension $d$. These are formalized as follows, again noting that the locations $A$ for a given term are finite and may be determined syntactically.

For a stack $S=\e\,\term{N_1}\dots\term{N_n}$ let $\term{[S]a.M}$ denote the term $\term{[N_1]a..[N_n]a.M}$, and for a memory $S_A$ let $\term{[S_A].M}$ denote a term $\term{[S_{a_1}]a_1..[S_{a_n}]a_n.M}$ where $A=\{a_1\dots a_n\}$. For a natural number $d$, a \emph{memory of dimension $d$} is one $S_A$ where $|S_A|=f$ with $f(a)\geq d$ for all $a\in A$, and each term in $S_A$ is a term of dimension $d-1$, where a \emph{term of dimension $d$} is of the form $\term{[S_A].*}$ with $S_A$ of dimension $d$. A \emph{substitution map of dimension $d$} is one that sends every variable to a term of dimension $d$. Observe that an element of dimension $d$ is also one of any $d'< d$.

\begin{restatable}[Spine normal forms terminate]{lemma}{spineNFterminate}
\label{lem:spine-normal => runs}
For a spine-normal term $\term\WW$ there is a dimension $d$ such that for any memory $S_A$ and substitution map $\sigma$ of dimension $d$, there is a memory $T_A$ such that $\eval {S_A}{\sigma\WW}{T_A}$.
\end{restatable}

Since machine termination guarantees typing (Theorem~\ref{thm:termination => types}), the above would be sufficient to connect spine normalization to weak typeability. The following lemma is however instructive in connecting spine reduction and machine termination explicitly.

\begin{restatable}[Evaluation is invariant under spine reduction]{lemma}{spineTermination}
\label{lem:spine reduction preserves termination}
If $\term{M}\rw_\spine\term{N}$ then $\eval{S_A}{\sigma M}{T_A}$ if and only if $\eval{S_A}{\sigma N}{T_A}$ for any memories $S_A$ and $T_A$ and substitution map $\sigma$.
\end{restatable}

The main theorem connecting weak typing, spine normalization, and machine termination is then the following.

\begin{theorem}[Weak types characterize termination and spine normalization]
\label{thm:weak}
For a term $\term M$ the following are equivalent:
\begin{enumerate}

\item\label{case:w} $\term M$ is weakly typeable $\term{G |w M:t}$,
\item\label{case:s} $\term M$ is spine-normalizing,
\item\label{case:t} $\term M$ is machine-terminating: there are memories $S_A$, $T_A$ and a substitution map $\sigma$ such that $\eval{S_A}{\sigma M}{T_A}$.

\end{enumerate}
\end{theorem}

\begin{proof}
By Theorem~\ref{thm:types <=> spine} statements~\ref{case:w} and~\ref{case:s} are equivalent. From statement~\ref{case:s} to~\ref{case:t}, if $\term M$ spine normalizes to $\term\WW$ then Lemma~\ref{lem:spine-normal => runs} gives memories $S_A$ and $T_A$ and a substitution map $\sigma$ such that $\eval{S_A}{\sigma\WW}{T_A}$, and since machine evaluation is reflected by spine reduction (Lemma~\ref{lem:spine reduction preserves termination}), also $\eval{S_A}{\sigma M}{T_A}$. From~\ref{case:t} to~\ref{case:w}, by Theorem~\ref{thm:termination => types} for $\eval{S_A}{\sigma M}{T_A}$ we get $\term{|w ??{S_A}{\sigma M}\e : \e=>[]^f}$ where $f=|T_A|$. The typing rule for states gives $\term{|w ?{S_A}:!!i}$ and $\term{|w \sigma M: !!i => []^f}$ for some $\type{!!i}$. By the weak substitution lemma (Lemma~\ref{lem:substitutionWeak}), since $\term{\sigma M}$ is typeable so is $\term M$, and we have $\term{G |w M: !!i => []^f}$ for some $\Gamma$.
\end{proof}

Note that the above theorem resolves the problem of uninhabited types as follows: if $\term M$ is typed $\term{G |w M:t}$ where either $\Gamma$ or $\type{t}$ contains an uninhabited (input) type, then (by~\ref{case:w} $\Rightarrow$~\ref{case:s}) it has a spine normal form, which in turn (by~\ref{case:s} $\Rightarrow$~\ref{case:t}) evaluates on the machine, which then (by~\ref{case:t} $\Rightarrow$~\ref{case:w}) guarantees that $\term M$ may also be typed with only inhabited types.

\section{Strong Normalization}
\label{sec:SN}

Where the weak type system characterizes machine evaluation and spine normalization, the strong type system characterizes strong normalization (SN). This section will prove one direction, that normalization reduces the weight of strong type derivations and, hence, that strong quantitative types guarantee SN and give upper bounds to the length of reductions. 

We start with the substitution lemma for strong types. The substitution typing rule can be eliminated, maintaining or reducing weight, and introduced, maintaining weight. The inequality is due to the weakening rule~$\RR W$: substitution into a weakened part of the context erases the typing information of a term, maintaining only the types for its context.

\begin{restatable}[Substitution, strong]{lemma}{substitutionStrong}
\label{lem:strong substitution}
If there are derivations $\term{G |s_n N:n}$ and $\term{x:n , D |s_m M:t}$ then there is a derivation $\term{G+D |s_k \{N/x\}M :t}$ where $k\leq n+m$, and vice versa where $k=n+m$.
\end{restatable}

The following lemma gives preservation of typing for beta-reduction, reducing weights.

\begin{restatable}[Subject reduction, beta]{lemma}{subjectReductionBeta}
\label{lem:beta subject reduction}
If $\term M\rw\term{N}$ by a beta-step and $\term{G |s_m M:t}$ then $\term{G |s_n N:t}$ where $m>n$.
\end{restatable}

While we have subject reduction, note that we do not have subject \emph{expansion}. In the reduction $\term{[N]a.a<x>.M}\rw\term{\{N/x\}M}$, if $x$ is not free in $\term M$, there is no typing information for $\term N$ in the reduct, but this is required for the strong application rule in the redex. This familiar situation presents an obstacle in demonstrating that all strongly normalising terms can be typed, covered in the next section.

For the remaining reduction rules, we do have subject expansion as well as reduction. Recall that the \emph{next} reduction, $\term{*;M}\rw\term M$, reduces weights, while the other non-beta reduction rules maintain weights. It would be possible to use an adjusted weight that counts \emph{only} application and abstraction typing rules, to give exact weights for non-beta subject reduction, since non-beta rules are inherently terminating by Proposition~\ref{prop:non-beta SN}.

\begin{restatable}[Subject reduction and expansion, non-beta]{lemma}{strongNonBeta}
\label{lem:non-beta subject reduction}
If $\term M\rw\term N$ by a non-beta-step then $\term{G |s_m M:t}$ if and only if $\term{G |s_n N:t}$ where $m\geq n$.
\end{restatable}

We conclude this section with the theorem that typed terms are strongly normalizing.

\begin{theorem}[Strong normalization]
\label{thm:SN}
If $\term{G |s M:t}$ then $\term M$ is strongly normalizing.
\end{theorem}

\begin{proof}
By \emph{subject reduction}, Lemmata~\ref{lem:beta subject reduction} and~\ref{lem:non-beta subject reduction}, reduction preserves types while maintaining or reducing weights. Since beta-steps strictly reduce the weight of a derivation and non-beta-steps are strongly normalizing by Proposition~\ref{prop:non-beta SN}, all reduction paths must be finite.
\end{proof}

\section{Perpetual Evaluation}
\label{sec:perp}

For the reverse direction, that all strongly normalizing terms are typeable, the challenge is that of \emph{weakened} terms, the argument terms $\term N$ in a reduction $\term{[N]a.a<x>.M}\rw\term{\{N/x\}M}$ where $x$ is not free in $\term M$. Reasoning by beta-expansion, to derive SN for $\term{[N]a.a<x>.M}$ we need termination information for $\term N$, which the reduct $\term{\{N/x\}M}=\term M$ does not immediately provide. This problem constitutes the essential difference between the $\lambda I$- and $\lambda K$-calculus (see~\cite{Barendregt-1984}) and is familiar from many strong normalization arguments, in particular those deriving SN from weak normalization by maintaining the weakened term $\term N$ in the term structure~\cite{Klop-1980,Nederpelt-1973} and those using perpetual reduction strategies~\cite{Bergstra-Klop-1982,Raamsdonk-Severi-Sorensen-Xi-1999}.

The present approach will be a variant of perpetual reduction. We will see a \emph{perpetual evaluation} relation $\term M\rp\term\MM$ which relates a term $\term M$ to its normal form $\term\MM$, that in the case of a redex $\term{[N]a.a<x>.M}$ not only evaluates $\term{\{N/x\}M}$ but in addition requires an evaluation $\term N\rp\term{\MM'}$. This definition, as a big-step evaluation relation instead of a strategy, avoids maintaining weakened terms in the syntax, either as redexes or explicit substitutions.

The proof then becomes straightforward: strong normalizability of $\term M$ implies perpetual evaluation $\term M\rp\term\MM$, implies strong typeability $\term{G |s M:t}$. As before, we begin by observing the shape of normal forms.

\begin{restatable}[Normal forms]{proposition}{normalforms}
The normal forms of $\rw$ are the terms $\term\MM$ given by the following grammars.
\[
\begin{array}{l@{~}l@{~}r@{~}l@{~}l}
	\term{\MM} &\Coloneqq&   \term{a<x>.\MM} &\mid& \term{\NN}
\\	\term{\NN} &\Coloneqq& \term{[\MM]a.\NN} &\mid& \term{x;\MM} ~\mid~ \term x ~\mid~ \term*
\end{array}
\qquad
\]
\end{restatable}

Perpetual evaluation will reduce terms to normal form by \emph{iterated weak head reduction}. That is, terms are reduced in a context of applications, a \emph{weak head context}, which functions similarly to the memory of the abstract machine. Where the machine would reach a \emph{final} or \emph{failure} state, reduction continues (i.e.\ is \emph{iterated}) for the remaining subterms, which on the machine would constitute the memory and continuation stack.

We will denote a term $\term M$ in weak head context $\term{[N_1]a_1..[N_n]a_n.M}$ as $\term{[S_A].M}$, where $S_A=\e\msep(\term{N_1})_{a_1}\cdots(\term{N_n})_{a_n}$ is the memory corresponding to the sequence of applications $\term{[N_1]a_1}$ through $\term{[N_n]a_n}$. Note that the notation $\term{[S_A]}$ is not a function on $S_A$, since multiple application sequences may correspond to the same memory (to be exact: those modulo $\term{[P]b.[N]a.M}\sim\term{[N]a.[P]b.M}$ where $a\neq b$). Typing, however, collapses these distinctions again: a term $\term{[S_A].*}$ will have the type $\type{\e => !!i}$ where $S_A:\type{!!i}$.

\begin{definition}[Perpetual evaluation]
\label{def:perpetual}
The \emph{perpetual evaluation} relation $\term M\rp\term\MM$ between terms and normal forms is given inductively by the rules in Figure~\ref{fig:perpetual}.
\end{definition}

\begin{figure}
\[
\begin{array}{@{}l@{\hspace{-25pt}}ccr@{}}
\text{Beta} & 
  \vc{\infer
    {\term{[S_A].[N]a.a<x>.M} \rp \term\MM}
    {\term{[S_A].\{N/x\}M} \rp \term\MM && \term N\rp\term{\MM'}} }
  & \vc{\infer{\term{[S_A].[N]b.a<x>.M}\rp\term\MM}{\term{[S_A].a<x>.[N]b.M}\rp\term\MM}}
  & \text{Passage}
\\ \\
\text{Next}
  & \vc{\infer{\term{[S_A].(*;M)} \rp \term\MM}{\term{[S_A].M}\rp \term\MM}}
  & \vc{\infer{\term{[S_A].(a<x>.N ; M)} \rp \term\MM}{\term{[S_A].a<x>.(N ; M)}\rp \term\MM}}
  & \text{Prefix (pop)}
\\ \\
\text{Associate}
  & \vc{\infer{\term{[S_A].((P;N) ; M)} \rp \term\MM}{\term{[S_A].(P ; (N;M))}\rp \term\MM}}
  & \vc{\infer{\term{[S_A].([P]a.N ; M)} \rp \term\MM}{\term{[S_A].[P]a.(N ; M)}\rp \term\MM}}
  & \text{Prefix (push)}
\\ \\
\text{Normal (abstraction)} &
	\multicolumn{3}{c}{\vc{\infer{\term{a<x>.M}\rp\term{a<x>.\MM}}{\term M\rp\term \MM}}}
\\ \\
\text{Normal (unit)} &
	\multicolumn{3}{c}{\vc{\infer{\term{[N_1]a_1..[N_n]a_n.*} \rp \term{[\MM_1]a_1..[\MM_n]a_n.*}}{(\term{N_i}\rp\term{\MM_i})_{i\leq n}}}}
\\ \\
\text{Normal (variable)} &
	\multicolumn{3}{c}{\vc{\infer{\term{[N_1]a_1..[N_n]a_n.x} \rp \term{[\MM_1]a_1..[\MM_n]a_n.x}}{(\term{N_i}\rp\term{\MM_i})_{i\leq n}}}}
\\ \\
\text{Normal (sequence)} &
  \multicolumn{3}{c}{\vc{\infer
    {\term{[N_1]a_1..[N_n]a_n.(x;M)} \rp \term{[\MM_1]a_1..[\MM_n]a_n.(x;\MM)}}
    {(\term{N_i}\rp\term{\MM_i})_{i\leq n} && \term M\rp\term\MM} }}
\end{array}
\]
\caption{Perpetual evaluation}
\label{fig:perpetual}
\end{figure}

The salient features of perpetual evaluation are the following. The first six rules each correspond to a reduction rule in weak head context. Simultaneously, these can be seen to implement machine evaluation, where the context represents the memory. The beta-rule requires the additional normalization of the argument $\term N$, as discussed above. The remaining four rules cover four of the five cases of normal forms, with the fifth, for application, covered by the presence of weak head contexts. As discussed above, these cases represent final or failure states on the machine, where reduction needs to iterate on the subterms representing the remaining memory and continuation stack.

The following lemma then shows that a strongly normalizing term also has a finite perpetual evaluation. To avoid reference to confluence, it is phrased as a combination of two statements: one, that SN implies perpetual evaluation, and two, that the normal form given by perpetual evaluation is indeed a normal form of the original term.

\begin{restatable}[Strong normalization implies perpetual evaluation]{lemma}{SNperpetual}
\label{lem:SN => perpetual}
If $\term M$ is strongly normalizing then $\term M\rp\term\MM$ and $\term M\rwn\term\MM$ for some $\term\MM$.
\end{restatable}

The definition of perpetual evaluation as an inductive relation means we may assign the same type to a term and its normal form simultaneously.

\begin{restatable}[Perpetual evaluation implies strong typeability]{lemma}{perpetualTyped}
\label{lem:perpetual => typed}
If $\term M\rp\term\MM$ then $\term{G |s M:t}$ and $\term{G |s \MM:t}$ for some $\Gamma$ and $\type t$.
\end{restatable}

This gives the required equivalence between strong quantitative typing, strong normalization, and perpetual evaluation.

\begin{theorem}[Strong typing characterizes strong normalization]
\label{thm:strong}
For a term $\term M$, the following are equivalent:
\begin{enumerate}

\item\label{SN:type} $\term M$ is strongly typeable, $\term{G |s M:t}$,
\item\label{SN:norm} $\term M$ is strongly normalizing,
\item\label{SN:perp} there is a perpetual evaluation $\term M\rp\term\MM$.
\end{enumerate}
\end{theorem}

\begin{proof}
From~\ref{SN:type} to~\ref{SN:norm} is Theorem~\ref{thm:SN}, from~\ref{SN:norm} to~\ref{SN:perp} is Lemma~\ref{lem:SN => perpetual}, and from~\ref{SN:perp} to~\ref{SN:type} is Lemma~\ref{lem:perpetual => typed}.
\end{proof}

\section{Quantitative Types for Higher-Order Store}
\label{sec:store}

The quantitative type system extends to computational effects and evaluation strategies via their encoding in the FMC. We will explore this by demonstrating that the present type system agrees with the intersection types of De'Liguoro and Treglia for a monadic $\lambda$-calculus with higer-order store~\cite{deLiguoro-Treglia-2021}. The comparison will be exact, but informal, for the reason that their type system is \emph{idempotent}, i.e.\ their type collections are sets, not multisets. Modulo this distinction, their typed calculus will embed directly into the FMC. Appendix~\ref{sec:QGM} in addition gives the correspondence with the related quantitative types of Alves, Kesner, and Ramos~\cite{Alves-Kesner-Ramos-2023}, a correspondence that is likewise informal but exact, as it is restricted to a fragment of their type system.

The untyped calculus consists of six syntactic constructs, which embed as follows, where the \emph{values} $V$ are the first column and the \emph{computations} $M$ the second and third. The embedding into the FMC uses the default location $\lambda$ for regular application, abstraction, and monadic constructs, omitted from the notation both in terms and types, i.e.\ $\term{[N].M}$ and $\term{<x>.M}$ abbreviate $\term{[N]l.M}$ and $\term{l<x>.M}$, and $\type{i}$ may abbreviate the singleton memory $\type{\lambda(i)}$.
\[
\begin{array}{@{}r@{}l@{\qquad}r@{}l@{\qquad}r@{}l@{}}
\multicolumn{2}{@{}l@{}}{\emph{Values $V$:}} & \multicolumn{2}{@{}l@{}}{\emph{Computations $M$:}}
\\                  x &~=~ \term{x}
&                 [V] &~=~ \term{[V].*}
& \get_a(\lambda x.M) &~=~ \term{a<x>.[x]a.M}
\\ \quad  \lambda x.M &~=~ \term{<x>.M}
&  \quad     M\bind V &~=~ \term{M;V}
&         \set_a(V,M) &~=~ \term{a<\_>.[V]a.M}
\end{array}
\]
The current presentation uses $M\bind V$ instead of De'Liguoro and Treglia's $M\star V$ to avoid conflicting with \emph{skip}, and $a,b,c\dots$ instead of $\ell$ for memory locations. The embeddings are standard~\cite{Douence-Fradet-1998,Heijltjes-2022,Power-Thielecke-1999}. They further agree with the encodings of application, let-syntax, and \emph{update} $a := M$ and lookup $!a$, as follows (the first column gives the encodings by De'Liguoro and Treglia~\cite{deLiguoro-Treglia-2021}, the second their standard FMC interpretation~\cite{Heijltjes-2022}).
\[
\begin{aligned}
                                V\,W &= [W]\bind V           &&&                    \term{[W].V} &\rwls \term{([W].*)\,;\,V}
\\ \mathsf{let}~x={}!a~\mathsf{in}~M &= \get_a(\lambda x.M)  &&& \term{(a<x>.[x]a.[x].*) ;<x>.M} &\rws  \term{a<x>.[x]a.M}
\\                      a := V\,;\,M &=   \set_a{V,M}        &&&       \term{(a<\_>.[V]a.*) ; M} &\rws  \term{a<\_>.[V]a.M}
\end{aligned}
\]
Types for the monadic calculus, given below, are stratified like those of the FMC, but unlike the latter they include an intersection operator $\wedge$ and unit $\omega$ at each level.
\[
\begin{array}{lr@{}l}
	\text{Value types:}         & \delta &~\Coloneqq~ \alpha  \mid  \delta\to\tau  \mid  \delta\wedge\delta'  \mid  \omega
\\  \text{Store types:}         & \sigma &~\Coloneqq~    \langle a:\delta \rangle  \mid  \sigma\wedge\sigma'  \mid  \omega
\\  \text{Configuration types:} & \kappa &~\Coloneqq~          \delta\times\sigma  \mid  \kappa\wedge\kappa'  \mid  \omega 
\\  \text{Result types:}        & \tau   &~\Coloneqq~             \sigma\to\kappa  \mid    \tau\wedge\tau'    \mid  \omega  
\end{array}
\]
A subtyping relation $\leq$ and equivalence $=$ for these types helps embed them into FMC types. First, implication distributes over intersection, $\delta\to(\tau\wedge\tau')=(\delta\to\tau)\wedge(\delta\to\tau')$, so that value types are given by a large intersection over arrow types $\delta\to\tau$ and variables $\alpha$. Value types then encode as collection types, where $\delta\wedge\delta'=\type{\delta + \delta'}$ and $\omega=\type{[]}$, and arrow types by their standard encoding~\cite{Heijltjes-2022}, where $\delta\to\tau$ adds the additional input $\delta$ to the FMC implication type given by $\tau$. Since value types are collections, $\delta\to\tau$ embeds as a singleton collection, as follows.
\[
	\delta\to\tau~=~\type{[\delta\,!!m => !!n]} \quad \text{where~}\tau~=~\type{!!m=> !!n}	
\]
For store types, collections distribute over location indexing, $\langle a,\delta\rangle\wedge\langle a,\delta'\rangle=\langle a,\delta\wedge\delta'\rangle$. A store type is thus of the following form, and embeds as a memory type over a set of locations $A$ representing (non-empty) memory cells.
\[
	\bigwedge_{a\in A}\langle a:\delta_a\rangle~=~\{\type{\delta_a}\mid a\in A\}
\]
Note that store types use the intersection operator to combine the types at different locations. Configuration types $\kappa$ are then memory types over a set of locations $A$ that include the default location $\lambda$, in the first line below, and computation types $\tau$ are implications from store types to memory types, second line below.
\[
\begin{aligned}
	\delta\times\sigma &~=~\type{!!i}\quad\text{where}~\type{!i_\lambda}=\type{\delta}~\text{and}~\type{!i_a}=\type{\sigma_a}~\text{for}~a\neq\lambda
\\	   \sigma\to\kappa &~=~\type{\sigma=>\kappa}
\end{aligned}
\]
Figure~\ref{fig:comparison} gives the comparison between intersection typing for store and quantitative FMC typing. With the difference that the former admits weakening and contraction, whereas for the latter contexts are linear, the relation is an embedding that exactly follows those of terms and types.

\begin{figure}
\[
\begin{array}{@{}cc@{}}
	\infer{\Gamma, x:\delta \vdash x:\delta}{}
&	\infer{\term{x:i |- x:i}}{\big(\infer{\term{x:[t] |- x:t}}{}\big)_{\type t\in\type i}}
\\ \\
	\infer{\Gamma\vdash \lambda x.M:\delta\ar\tau}{\Gamma, x:\delta\vdash M:\tau}
&	\infer{\term{G |- <x>.M : i\,!!m => !!n}}{\term{G, x:i |- M : !!m=>!!n}}
\\ \\
	\infer{\Gamma\vdash [V]:\sigma\ar\delta\tim\sigma}{\Gamma\vdash V:\delta}
&	\infer{\term{G |- [V].* : !!m => !!m\,i}}{\term{G |- V : i} && \infer{\term{|- * : i\,!!m => !!m\,i}}{}}
\\ \\
	\infer{\Gamma\vdash M\star V : \sigma \ar \delta''\tim\sigma''}{\Gamma\vdash M : \sigma \ar \delta'\tim\sigma' & \Gamma \vdash V : \delta'\ar\sigma'\ar\delta''\tim\sigma''} 
&	\infer{\term{G+D |- M;V : !!m => !!m" \,i" }}{\term{G |- M : !!m => !!m'\,i'} & \term{D |- V : i'\,!!m' => !!m"\,i"}}
\\ \\
	\infer{\Gamma \vdash \get_a(\lambda x.M) : (\langle a:\delta\rangle\wedge \sigma)\ar\kappa}{\Gamma , x:\delta \vdash M : \sigma \ar \kappa}
&   \infer{\term{G |- a<x>.[x]a.M : a(i+i')\,!!m => !!n}}{
     \infer{\term{G , x:i+i' |- [x]a.M : !!m => !!n}}{
       \infer{\term{x:i |- x:i}}{\big(\infer{\term{x:[t]|-x:t}}{}\big)_{\type t\in\type i}}
     & \term{G , x:i' |- M : a(i)\,!!m => !!n}
    }}
\\ \\
	\infer{\Gamma \vdash \set_a(V,M) : \sigma \ar \kappa}{\Gamma \vdash V : \delta && \Gamma \vdash M : (\langle a:\delta \rangle\wedge\sigma)\ar \kappa}
&	\infer{\term{G+D |- a<\_>.[V]a.M : a([])\,!!m => !!n}}{\infer{\term{G+D |- [V]a.M : !!m => !!n}}{\term{G |- V:i} && \term{D |- M : a(i)\,!!m => !!n}}}
\end{array}
\]
\caption{A comparison of intersection typing for store to quantitative FMC typing. Let $\delta=\type i$, $\sigma=\type{!!m}$, $\kappa=\type{!!n}$, and $\tau=\type{!!m=>!!n}$, and similarly for $\delta'$, $\delta''$, etc. The last rule on the left has the side condition that $a$ is not assigned in $\sigma$, which gives the correspondence between $\sigma$ and $\type{a([])\,!!m}$ on the right via the subtyping inequality $\omega \leq \langle a:\omega\rangle$, by $\sigma = \omega\wedge\sigma \leq \langle a:\omega\rangle\wedge\sigma = \type{a([])\,!!m}$.
}
\label{fig:comparison}
\end{figure}

\section{Conclusion}

Quantitative types for the FMC arise from simple types by the standard generalization, of replacing input types (or \emph{value} types) with multisets. Nevertheless, through the encodings of effects and strategies in the FMC, this gives a quantitative characterization of effectful higher-order programs, in agreement with known results.

Of the technical contributions of the paper, two deserve closer attention. One is the exact quantification of machine evaluation by the weak type system, simply by counting the typing rules corresponding to machine steps, without the need for dedicated technical adaptations. Fundamentally this relies on the extension of the Krivine machine with sequencing and, in particular, \emph{skip}, as signifying successful termination. The direct correspondence between weak quantitative types and machine evaluation underlines how this is an interesting and natural extension of the $\lambda$-calculus that broadens our perspective on higher-order computation.

The second main technical contribution is the perpetual evaluation relation. Its equivalence with strong normalization holds independently of the strong type system, and it is expected to prove useful in strong normalization proofs for potential future type systems.


\bibliographystyle{plainurl}
\bibliography{FMC}



\newpage
\appendix


\newcommand\aaa{\mathrm{a}}
\newcommand\vvv{\mathrm{v}}
\newcommand\ttt{\mathrm{tt}}
\newcommand\nnn{\mathrm{n}}

\newcommand\MMM{\mathcal M}
\newcommand\NNN{\mathcal N}
\newcommand\SSS{\mathcal S}

\section{Addendum to Section~\ref{sec:store}: encoding \emph{Quantitative Global Memory}}
\label{sec:QGM}

This Appendix will compare the present quantitative types against the multi-types of Alves, Kesner, and Ramos in \emph{Quantitative Global Memory}~\cite{Alves-Kesner-Ramos-2023}. The paper uses the same calculus as De'Liguoro and Treglia~\cite{deLiguoro-Treglia-2021}, characterising it as \emph{weak open call--by--value}, i.e.\ not reducing under abstractions while allowing open terms. The calculus and its FMC encoding are as follows, using the standard interpretation of call--by--value $\lambda$-calculus~\cite{Douence-Fradet-1998,Heijltjes-2022,Power-Thielecke-1999}, where the standard encoding of call--by--value application, $v\,t = \term{t ; [v] ; <x>.x}$, is reduced to $\term{t ; v}$.
\[
\begin{array}{@{}r@{}l@{\qquad}r@{}l@{\qquad}r@{}l@{}}
\multicolumn{2}{@{}l@{}}{\emph{Values $v$:}} & \multicolumn{2}{@{}l@{}}{\emph{Terms $t$:}} & 
\\                    x &~=~ \term{x}      &          v &~=~ \term{[v].*} & \get_a(\lambda x.t) &~=~ \term{a<x>.[x]a.t}
\\ \quad    \lambda x.t &~=~ \term{<x>.t}  & \quad v\,t &~=~ \term{t ; v} & \set_a(v,t)         &~=~ \term{a<\_>.[v]a.t}
\end{array}
\]
The paper gives a system of quantitative types for this calculus. Here, we will consider it in simplified form, omitting base types and restricting our focus to four (out of eight) categories. The type system, and its encoding in quantitative FMC types, is then as follows.
\[
\begin{array}{lr@{}l@{}l}
   \text{Value types:}         & \sigma &~\Coloneqq~ \MMM\smallbin\Rightarrow(\SSS\then\kappa) &~=~ \type{\lambda(\MMM)\,\SSS => \kappa}
\\ \text{Configuration types:} & \kappa &~\Coloneqq~ \MMM \tim \SSS                            &~=~ \type{\lambda(\MMM)\,\SSS}
\\ \text{Multi-types:}         & \MMM   &~\Coloneqq~ [\sigma_1,\dots,\sigma_n]                 &~=~ \type{[\sigma_1,\dots,\sigma_n]}
\\ \text{State types:}         & \SSS   &~\Coloneqq~ \{a_1:\MMM_1,\dots,a_n:\MMM_n\}           &~=~ \type{a_1(\MMM_1)\dots a_n(\MMM_n)}
\end{array}
\]
That is: \emph{value types} encode as computation types $\type{t}$, \emph{state types} as memory types $\type{!!i}$ excluding the main location $\lambda$, \emph{multi-types} as collection types $\type{i}$, and \emph{configuration types} as memory types $\type{!!i}$.


The type system features the properties that it is \emph{tight}, it counts exactly the reduction steps to normal form, and it is \emph{split}, it separately counts beta-steps, store operations, and the size of the normal form. These quantities are however not meaningful in the FMC encoding, which measures machine evaluation rather than reduction, identifies beta-reduction and store operations (both are stack interactions), and considers spine reduction, while the encoded call--by--value reduction corresponds to weak head reduction in the FMC. Accordingly, we will here consider only the structural correspondence between both type systems, omitting the triple counters and the base types $\aaa$, $\vvv$, and $\nnn$ which are used to identify normal forms. Thus simplified, the type system and its FMC encoding are given in Figure~\ref{fig:comparison-2}.

\begin{figure}
\[
\begin{array}{@{}c@{}c@{}}
    \infer{x:[\sigma] \vdash x:\sigma}{}
&   \infer{\term{x:[\sigma] |- x:\sigma}}{}
\\ \\ 
    \infer{\Gamma \vdash \lambda x.t : \MMM \smallbin\Rightarrow (\SSS \then \kappa)}{\Gamma , x:\MMM \vdash t:\SSS \then \kappa}
&   \infer{\term{G |- <x>.t : \MMM\,\SSS => \kappa}}{\term{G , x:\MMM |- t : \SSS => \kappa}}
\\ \\
	\infer{\Gamma_1+\dots+\Gamma_n \vdash v : [\sigma_1,\dots,\sigma_n]}{(\Gamma_i\vdash v:\sigma_i)_{i\leq n}}
&   \infer{\term{G_1++G_n |- v : [\sigma_1,,\sigma_n]}}{(\term{G_i |- v:\sigma_i})_{i\leq n}}
\\ \\
	\infer{\Gamma \vdash v:\SSS\then(\MMM\tim\SSS)}{\Gamma \vdash v:\MMM}
&	\infer{\term{G |- [v].* : \SSS => \MMM\,\SSS}}{\term{G |- v:\MMM} && \infer{\term{|- * : \MMM\,\SSS => \MMM\,\SSS}}{}}
\\ \\
	\infer{\Gamma+\Delta \vdash v\,t:\SSS\then\kappa}{\Gamma\vdash v:\MMM\smallbin\Rightarrow(\SSS'\then\kappa) & \Delta\vdash t:\SSS\then(\MMM\tim\SSS')}
&	\infer{\term{G+D |- t;v : \SSS => \kappa}}{\term{D |- t: \SSS => \MMM\,\SSS'} && \term{G |- v: \MMM\,\SSS' => \kappa}}
\\ \\
	\infer{\Gamma\vdash \get_a(\lambda x.t):\{a:\MMM+\NNN\};\SSS\then\kappa}{\Gamma,x:\MMM\vdash t:\{a:\NNN\};\SSS\then\kappa}
&	\infer{\term{G |- a<x>.[x]a.t : a(\MMM+\NNN)\,\SSS => \kappa}}{\infer{\term{G , x:\MMM+\NNN |- [x]a.t : \SSS => \kappa}}{\infer{\term{x:\NNN |- x:\NNN}}{\big(\infer{\term{x:[s]|-x:s}}{}\big)_{\sigma\in\NNN}} && \term{G , x:\MMM |- t: a(\NNN)\,\SSS => \kappa}}}
\\ \\
	\infer{\Gamma+\Delta \vdash \set_a(v,t): {\color{green!80!black}\{a:[]\};{}}\SSS\then\kappa}{\Gamma\vdash v:\MMM & \Delta\vdash t: \{a:\MMM\};\SSS\then\kappa}
&   \infer{\term{G+D |- a<\_>.[v]a.t : a([])\,\SSS => \kappa}}{\infer{\term{G+D |- [v]a.t : \SSS => \kappa}}{\term{G |- v:\MMM} && \term{D |- t: a(\MMM)\,\SSS => \kappa}}}
\end{array}
\]
\caption{A comparison of multi-types for store to quantitative FMC typing. The FMC types omit the default location $\lambda$, using $\type{\MMM\,\SSS}$ for $\type{\lambda(\MMM)\,\SSS}$. In the last rule on the left, the addition of $\{a:[]\}$ (in green) to the conclusion type corrects a minor mistake in~\cite{Alves-Kesner-Ramos-2023} (confirmed in private communication with the authors).}
\label{fig:comparison-2}
\end{figure}

\section{Proofs for Section~\ref{sec:FMC}: The Functional Machine Calculus}

\smallStepBigStep*

\begin{proof}
The machine transitions correspond to the big-step evaluation rules as follows. From left to right the proof is by (strong) induction on $n$, and from right to left by induction on the derivation for $\evalarrow_n$.

\begin{itemize}
\item\case{Application:}
A successful run starting with a \emph{push} transition and the evaluation rule for application correspond as follows.
\[
\begin{array}{@{(~}l@{~,~}r@{~,~}r@{~)}}
          S_A               & \term{[N]a.M} & \e
\\ \hline S_A\msep\term N_a &      \term{M} & \e
\\ \dline T_A               &      \term{*} & \e
\end{array}
\,\raisebox{-4pt}{$\scriptstyle n$}
\qquad
\evalAlign 
  [n]   {S_A\msep\term N_a}       M  {T_A}
  [n+1] {S_A}               {[N]a.M} {T_A}
\]

\item\case{Abstraction:}
A successful run starting with a \emph{pop} transition and the evaluation rule for abstraction correspond as follows.
\[
\begin{array}{@{(~}l@{~,~}r@{~,~}r@{~)}}
          S_A\msep\term N_a & \term{a<x>.M} & \e
\\ \hline S_A               &      \term{M} & \e
\\ \dline T_A               &      \term{*} & \e
\end{array}
\,\raisebox{-4pt}{$\scriptstyle n$}
\qquad  
\evalAlign 
  [n]   {S_A}             {\{N/x\}M} {T_A}
  [n+1] {S_A\msep\term N_a} {a<x>.M} {T_A}
\]

\item\case{Unit:}
The unit run for a single state corresponds to the unit evaluation rule.
\[
	(S_A,\term *,\e) \qquad \evalRule {} {\eval[1]{S_A}*{S_A}}
\]

\item\case{Sequence:}
A successful run for $\term{M;N}$ consists of a transition pushing $\term N$ to the continuation stack, a successful run for $\term M$ (ignoring $\term N$ on the continuation stack), a transition popping $\term N$, and a run for $\term N$. The correspondence is then as follows, where the total length of the run is that of both constituting runs for $\term M$ and $\term N$ plus the initial state for $\term{M;N}$.
\[
\begin{array}{@{(~}l@{~,~}r@{~,~}r@{~)}l@{}}
          R_A & \term{M;N} & \e       
\\ \hline R_A & \term{M}   & \term{N} 
\\ \dline S_A & \term{*}   & \term{N} 
\\ \hline S_A & \term{N}   & \e       
\\ \dline T_A & \term{*}   & \e       
\end{array}
\,\raisebox{11pt}{\makebox[0pt][l]{$\scriptstyle m$}}
  \raisebox{-20pt}{$\scriptstyle n$}
\qquad
\evalRule {\eval[m]{R_A}M{S_A} \qquad \eval[n]{S_A}N{T_A}} {\eval[m+n+1]{R_A}{M;N}{T_A}}
\]
\end{itemize}
\end{proof}

\section{Proofs for Section~\ref{sec:machine termination}: Machine Termination}

\substitutionWeak*

\begin{proof}
The proof is by induction on $\term M$, except that the \emph{collection} case is treated separately from the \emph{application} case. Each case is demonstrated by two derivations, representing a permutation of the substitution typing rule, where one derivation exists if and only if the other does.

\begin{itemize}

\item\case{Variable case:}
In the case $\term M=\term x$ we have the following two derivations. In both directions it is immedate that given one derivation, we may construct the other.
\[
\infer-[\rr s]{\term{G |- \{N/x\}x:t}}{\infer[\rr c]{\term{G |- N:[t]}}{\term{G |- N:t}} && \infer[\rr x]{\term{x:[t] |- x:t}}{}}
\qquad\qquad
\term{G |- N:t}
\]
In the case $\term M=\term y\neq\term x$ we have the following two derivations.
\[
\infer-[\rr s]{\term{y:[t] |- \{N/x\}y:t}}{\infer[\rr c]{\term{|- N:[]}}{} && \infer[\rr x]{\term{y:[t], x:[] |- y:t}}{}}
\qquad\qquad
\infer[\rr x]{\term{y:[t] |- y:t}}{}
\]

\item\case{Abstraction case:}
In the case $\term M=\term{a<y>.P}$ we have the following two derivations. In both directions, given one derivation we may construct the other, where by $\alpha$-conversion we may assume $\term y\notin\fv N$.
\[
\infer-[\rr s]{\term{G+D |- \{N/x\}a<y>.P: a(i)\,!!k => !!m }}{\term{G |- N:n} && \infer[\rr l]{\term{D , x:n |- <y>.P: a(i)\,!!k => !!m}}{\term{D , x:n , y:i , |- P: !!k => !!m }}}
\]
\[
\infer[\rr l]{\term{G+D |- a<y>.\{N/x\}P: a(i)\,!!k => !!m}}{\infer-[\rr s]{\term{G+D, y:i |- \{N/x\}P: !!k => !!m}}{\term{G |- N:n} && \term{D , x:n , y:i , |- P: !!k => !!m }}}
\]

\item\case{Weak application case:}
In the case $\term M=\term{[Q]a.P}$ we have the following two derivations, where $\type n =\type{n_1+n_2}$ and $\term D=\term{D_1+D_2}$. Given the former derivation, since $\term{G |- N:n}$ must be introduced by a collection rule, we have $\term{G_i |- N:n_i}$ for $i=1,2$ with $\term G=\term{G_1+G_2}$, so that we may construct the latter derivation. The other direction is immediate.
\[
\infer-[\rr s]{\term{G+D |- \{N/x\}[Q]a.P: !!k => !!m}}{
   \term{G |- N:n}
&& \infer[\rr a]{\term{D , x:n |w [Q]a.P}}{
      \term{D_1 , x:n_1 |- Q:i}
   && \term{D_2 , x:n_2 |- P:a(i)\,!!k => m}
   }
}
\]
\[
\infer[\rr a]{\term{G+D |w [\{N/x\}Q]a.\{N/x\}P: !!k => !!m}}{
   \infer-[\rr s]{\term{G_1+D_1 |- \{N/x\}Q : i}}{
      \term{G_1 |- N:n_1}
   && \term{D_1 , x:n_1 |- Q:i}
   }
&& \infer-[\rr s]{\term{G_2+D_2 |- \{N/x\}P : a(i)\,!!k => !!m}}{
      \term{G_2 |- N:n_2}
   && \term{D_2 , x:n_2 |- P:a(i)\,!!k => !!m}
   }
}
\]

\item\case{Unit case:} We have the following two derivations.
\[
\infer-[\rr s]{\term{|- \{N/x\}* : !!i => !!i}}{\infer[\rr c]{\term{|- N:[]}}{} && \infer[\rr *]{\term{|- *: !!i => !!i}}{}}
\qquad\qquad
\infer[\rr *]{\term{|- *: !!i => !!i}}{}
\]

\item\case{Sequencing case:} The case is similar to the weak application case, with the following two derivations. Given the former derivation, since $\type n=\type{n_1+n_2}$ and $\term{G |- N:n}$ is introduced by an \emph{intersection} rule, we have the two derivations $\term{G_i |- N:n_i}$ for $i=1,2$, and may construct the latter derivation. The reverse is again immediate, noting that $\term G=\term{G_1+G_2}$ and $\term D=\term{D_1+D_2}$.
\[
\infer-[\rr s]{\term{G+D |- \{N/x\}(P;Q) : !!i => !!m}}{
   \term{G |- N:n}
&& \infer[\rr ;]{\term{D , x:n |- P;Q : !!i => !!m}}{
      \term{D_1 , x:n_1 |- P : !!i => !!k}
   && \term{D_2 , x:n_2 |- Q : !!k => !!m}
   }
}  
\]
\[
\infer[\rr ;]{\term{G+D |- \{N/x\}(P;Q) : !!i => !!m}}{
   \infer-[\rr s]{\term{G_1+D_1 |- \{N/x\}P : !!i => !!k}}{
      \term{G_1 |- N:n_1}
   && \term{D_1 , x:n_1 |- P: !!i => !!k}
   }
&& \infer-[\rr s]{\term{G_2+D_2 |- \{N/x\}Q : !!k => !!m}}{
      \term{G_2 |- N:n_2}
   && \term{D_2 , x:n_2 |- Q: !!k => !!m}
   }
}
\]

\item\case{Collection case:} The case is similar to the weak application case. The two derivations are as follows, where $\type n=\type{n_1++n_n}$ and $\term D=\term{D_1++D_n}$. Given the former derivation, since $\term{G |- N:n}$ is introduced by an intersection rule, we have the derivations $\term{G_i |- N:n_i}$ for $1\leq i\leq n$ needed to construct the latter derivation, where $\term G=\term{G_1++G_n}$. The reverse direction is immediate.
\[
\infer-[\rr s]{\term{G+D |- \{N/x\}M : [t_1,,t_n]}}{
   \term{G |- N:n}
&& \infer[\rr c]{\term{D , x:n |- M: [t_1,,t_n]}}{
      (\term{D_i , x:n_i |- M: t_i})_{1\leq i\leq n}
   }
}
\]
\[
\infer[\rr c]{\term{G+D |- \{N/x\}M : [t_1,,t_n]}}{
   \hspace{22pt}\left(\vc{
   \infer-[\rr s]{\term{G_i+D_i |- \{N/x\}M : t_i}}{
      \term{G_i |- N:n_i}
   && \term{D_i , x:n_i |- M:t_i}
   }
   }\right)_{1\leq i\leq n}
   \hspace{34pt}
}
\qedhere
\]
\end{itemize}
\end{proof}

\evaluationReducesWeights*

\begin{proof}
For each of the four machine transitions the two typing derivations will be given. In each case, given one derivation the other may be constructed, while the first has weight exactly one greater than the second, as it has exactly one additional rule instance contributing to the weight of a derivation (an \emph{application}, \emph{abstraction}, \emph{unit}, or \emph{sequencing} rule).

\begin{itemize}
	
\item\case{Application case:}
\[
	\step {S_A}{[N]a.M}K {S_A\msep\term N_a}MK
\]
The two typing derivations are as follows.
\[
	\infer[\rr ?]{\term{|- ??{S_A}{[N]a.M}K : \e => !!m }}{
	   \term{|- S_A : !!i}
	&& \infer[\rr a]{\term{|- [N]a.M : !!i => !!k}}{\term{|- N : n} && \term{|- M : a(n)\,!!i => !!k}}
	&& \term{|- K : !!k => !!m}
	}
\]
\[
	\infer[\rr ?]{\term{|- ??{S_A\msep\term N_a}MK : \e => !!m }}{
	   \infer[\rr o]{\term{|- ?{S_A\msep\term N_a} : !!i\,a(n)}}{\term{|- ?{S_A} : !!i} && \term{|- N : n}}
	&& \term{|- M : a(n)\,!!i => !!k}
	&& \term{|- K : !!k => !!m}
	}
\]

\item\case{Abstraction case:}
\[
	\step{S_A\msep\term N_a}{a<x>.M}K {S_A}{\{N/x\}M}K
\]
The two typing derivations are as follows. In both directions, the use of the \emph{substitution} typing rule is justified by Lemma~\ref{lem:substitutionWeak}.
\[
	\infer[\rr ?]{\term{|- ??{S_A\msep\term N_a}{a<x>.M}K : \e => !!m}}{
	   \infer[\rr o]{\term{|- ?{S_A\msep\term N_a} : !!i\,a(n)}}{\term{|- ?{S_A} : !!i} && \term{|- N : n}}
	&& \infer[\rr l]{\term{|- a<x>.M : a(n)\,!!i => !!k}}{\term{x:n |- M : !!i => !!k}}
	&& \term{|- K : !!k => !!m}
	}
\]
\[
	\infer[\rr ?]{\term{|- ??{S_A}{\{N/x\}M}K : \e => !!m}}{
	   \term{|- ?{S_A} : !!i}
	&& \infer-[\rr s]{\term{|- \{N/x\}M : !!i => !!k}}{\term{|- N : n} && \term{ x:n |- M : !!i => !!k}}
	&& \term{|- K : !!k => !!m}
	}
\]

\item\case{Sequencing case:}
\[
	\step {S_A}{M;N}K {S_A}M{\term N\,K} 
\]
The two derivations are as follows.
\[
	\infer[\rr ?]{\term{|- ??{S_A}{M;N}K : \e => !!n}}{
	   \term{|- ?{S_A} : !!i}
	&& \infer[\rr ;]{\term{|- M;N : !!i => !!m}}{\term{|- M : !!i => !!k} && \term{|- N : !!k => !!m}}
	&& \term{|- ?K : !!m => !!n}
	}
\]
\[
	\infer[\rr ?]{\term{|- ??{S_A}M{\term N\,K} : \e => !!n}}{
	   \term{|- ?{S_A} : !!i}
	&& \term{|- M : !!i => !!k}
	&& \infer[\rr k]{\term{|- ?{\term N\,K} : !!k => !!n}}{ 
	       \term{|- N : !!k => !!m} 
	   &&  \term{|- ?K : !!m => !!n}
	   }
	}
\]

\item\case{Unit case:}
\[
	\step {S_A}*{\term M\,K} {S_A}MK
\]
The two derivations are as follows.
\[
    \infer[\rr ?]{\term{|- ??{S_A}*{\term M\,K} : \e => !!m}}{
    	\term{|- ?{S_A} : !!i}
    &&	\infer[\rr *]{\term{|- * : !!i => !!i}}{}
    &&	\infer[\rr k]{\term{|- ?{\term M\,K} : !!i => !!m}}{ 
			\term{|- M : !!i => !!k} 
		&&	\term{|- ?K : !!k => !!m } 
		}
	}
\]
\[
	\infer[\rr ?]{\term{|- ??{S_A}MK : \e => !!m}}{
		\term{|- ?{S_A} : !!i}
	&&	\term{|- M : !!i => !!k}
	&&	\term{|- ?K : !!k => !!m}
	}
\qedhere
\]
\end{itemize}
\end{proof}

\section{Proofs for Section~\ref{sec:spine}: Spine Normalization}

\spineNF*

\begin{proof}
By induction on terms, it is shown that $\term{\WW}$ are the spine normal terms, and $\term{\VV}$ those that are not abstractions. A variable $\term x$ and a unit $\term *$ are spine normal, and a composition $\term{M;N}$ is normal if and only if $\term M$ is a variable and $\term N$ is spine normal, since if $\term M$ is a unit, application, abstraction, or sequencing, a reduction step applies. An application $\term{[N]a.M}$ is spine-normal if and only if $\term M$ is spine-normal and not an abstraction, since otherwise a beta-step or passage-step would apply. An abstraction $\term{a<x>.M}$ is spine-normal if and ony if $\term M$ is. 
\end{proof}

\spineNFTyped*

\begin{proof}
By induction on $\term{\WW}$ the following are proved.
\begin{enumerate}
	\item\label{case:N0} For any $\term{\WW}$ there is a type derivation $\term{G |w \WW: !!i => []^f}$.
	\item\label{case:N1} For any $\term{\VV}$ and any $f:A\to\mathbb N$ there is a type derivation $\term{G |w \VV: []^f => []^g}$.
\end{enumerate}
Note that~\ref{case:N1} implies~\ref{case:N0}.

\begin{itemize}

\item\case{Unit case:}
For a term $\term *$, we prove case~\ref{case:N1}. For any $f$ we have $\term{|- *:[]^f => []^f}$.

\item\case{Variable case:}
For a variable $\term x$, we prove case~\ref{case:N1}. For any $f$ we have $\term{x:[[]^f=>\e] |- x: []^f => \e}$.

\item\case{Sequencing case:}
For a term $\term{x;\WW}$ we prove case~\ref{case:N1}. By induction, let $\term{G |- \WW: !!i => []^g}$. For any $f$ we get $\term{G + `{\black(}x:[]^f=>!!i`{\black)} |- x;\WW : []^f => []^g}$.

\item\case{Application case:}
For a term $\term{[M]a.\VV}$, we prove case~\ref{case:N1}. Given $f$, define $f'$ by $f'(a)=f(a)+1$ and $f'(b)=f(b)$ for $a\neq b$. We apply the inductive hypothesis for case~\ref{case:N1} to $\term{\VV}$ and $f'$ to get $\term{G |- \VV: []^{f'} => []^g}$. The following \emph{weak application} rule instance then gives the required type.
\[
\infer[\rr a]{\term{G |w [M]a.\VV : []^f => []^g}}{
   \infer[\rr c]{\term{|- M: []}}{}
&& \term{G |- \VV: []^{f'} => []^g}
}
\]

\item\case{Abstraction case:}
For a term $\term{a<x>.\WW}$ we prove case~\ref{case:N0}. By induction, there is a judgement $\term{G , x:i |- \WW : !!k => []^f}$. Then by the \emph{abstraction} typing rule we get the required type $\term{G |- a<x>.\WW : a(i)\,!!k => []^f}$.
\qedhere
\end{itemize}
\end{proof}

\weightedSpine*

\begin{proof}
By induction on spine contexts. The cases for a top-level reduction step are as follows.

\begin{itemize}

\item\case{Beta case:} For a reduction $\term{[N]a.a<x>.M}\rw\term{\{N/x\}M}$, by Lemma~\ref{lem:substitutionWeak} a derivation for $\term{\{N/x\}M}$ is given by a substitution rule, as below right. Then we have the derivation below left if and only if we have that below right, with weight $n+2$ respectively $n$.
\[
\infer[\rr a]{\term{G+D |w [N]a.a<x>.M : !!k => !!m}}{
   \term{G |- N:i}
&& \infer[\rr l]{\term{D |- a<x>.M: a(i)\,!!k => !!m}}{
      \term{D , x:i |- M: !!k => !!m}
   }
}
\qquad
\infer-[\rr s]{\term{G+D |- \{N/x\}M : !!k => !!m}}{
   \term{G |- N:i}
&& \term{D , x:i |- M: !!k => !!m}
}
\]

\item\case{Passage case:} For a reduction $\term{[N]b.a<x>.M}\rw\term{a<x>.[N]b.M}$ where $a\neq b$ and $x\notin\fv N$, we have the derivation below left if and only if we have that below right, with equal weights. Note that $\type{a(i)\,b(n)\,!!k}=\type{b(n)\,a(i)\,!!k}$.
\[
\infer[\rr a]{\term{G+D |w [N]b.a<x>.M : a(i)\,!!k => !!m}}{
   \term{G |- N:i}
&& \infer[\rr l]{\term{D |- a<x>.M: a(i)\,b(n)\,!!k => !!m}}{
      \term{D , x:i |- M: b(n)\,!!k => !!m}
   }
}
\qquad
\infer[\rr a]{\term{G+D |w a<x>.[N]b.M : a(i)\,!!k => !!m}}{
   \infer[\rr l]{\term{G+D , x:i |- [N]b.M: !!k => !!m}}{
      \term{G |- N:i}
   && \term{D , x:i |- M: b(n)\,!!k => !!m}
   }
}
\]

\item\case{Next case:} For a reduction $\term{*;M}\rw\term M$, we have the derivation for $\term{G |_{n+2} *;M : !!i => !!k}$ below if and only if we have a derivation for $\term{G |_n M: !!i => !!k}$.
\[
\infer[\rr ;]{\term{G |- *;M : !!i => !!k}}{
   \infer[\rr *]{\term{|- *: !!i => !!i}}{}
&& \term{G |- M: !!i => !!k}
}
\]

\item\case{Prefix (pop) case:} For a reduction $\term{a<x>.N;M}\rw\term{a<x>.(N;M)}$ where $x\notin\fv M$, given the derivation below left we may construct that below right, and vice versa.
\[
\infer[\rr ;]{\term{G+D |- a<x>.N;M : a(i)\,!!k => !!n }}{
   \infer[\rr l]{\term{G |- a<x>.N : a(i)\,!!k => !!m}}{\term{G , x:i |- N: !!k => !!m}}
&& \term{D |- M : !!m => !!n}   
}
\quad~
\infer[\rr l]{\term{G+D |- a<x>.(N;M) : a(i)\,!!k => !!n }}{
   \infer[\rr ;]{\term{G+D , x:i |- N;M : !!k => !!n }}{
      \term{G , x:i |- N: !!k => !!m}
   && \term{D |- M : !!m => !!n}   
}}
\]

\item\case{Prefix (push) case:} For a reduction $\term{[P]a.N;M}\rw\term{[P]a.(N;M)}$, given the first derivation below we may construct the second, and vice versa.
\[
\infer[\rr ;]{\term{G+D+E |- [P]a.N;M : !!k => !!n }}{
   \infer[\rr a]{\term{G+D |w [P]a.N : !!k => !!m}}{
      \term{G |- P: i}
   && \term{D |- N: a(i)\,!!k => !!m}
   }
&& \term{E |- M : !!m => !!n}   
}
\]
\[
\infer[\rr a]{\term{G+D+E |w [P]a.(N;M) : !!k => !!n }}{
   \term{G |- P: i}
&& \infer[\rr ;]{\term{D+E |- N;M : a(i)\,!!k => !!n}}{
      \term{D |- N: a(i)\,!!k => !!m}
   && \term{E |- M : !!m => !!n}   
}}
\]

\item\case{Associate case:} For a reduction $\term{(P;N);M}\rw\term{P;(N;M)}$, given the first derivation below we may construct the second, and vice versa.
\[
\infer[\rr ;]{\term{G+D+E |- (P;N);M : !!i => !!n }}{
   \infer[\rr ;]{\term{G+D |- P;N : !!i => !!k}}{
      \term{G |- P: !!i => !!k}
   && \term{D |- N: !!k => !!m}
   }
&& \term{E |- M : !!m => !!n}   
}
\]
\[
\infer[\rr ;]{\term{G+D+E |- P;(N;M) : !!i => !!n }}{
   \term{G |- P: !!i => !!k}
&& \infer[\rr ;]{\term{D+E |- N;M : !!k => !!n}}{
      \term{D |- N: !!k => !!m}
   && \term{E |- M : !!m => !!n}   
}}
\]
\end{itemize}

For reduction in a spine context, $\term{W\{M\}}\rw_\spine\term{W\{N\}}$, it is sufficient to observe from the derivation rules that the  derivation for $\term M$ occurs exactly once in $\term{W\{M\}}$; then the effect of replacing the derivation for $\term{G |w M:t}$ with $\term{G |w N:t}$, or vice versa, reflects directly on the derivations for $\term{W\{M\}}$ and $\term{W\{N\}}$.
%
%
%
%
%
%
%
\end{proof}

\spineNFterminate*

\begin{proof}
By induction on $\term\WW$ it will be shown that there is a dimension $d$ such that for any memory $S_A$ and substitution map $\sigma$ of dimension $d$, the following hold.
\begin{enumerate}

\item\label{case:spine eval 1} 
There is a memory $T_A$ such that $\eval{S_A}{\sigma\WW}{T_A}$.

\item\label{case:spine eval 2}
For a term $\term\VV$, there is a memory $T_A$ such that $\eval\e{\sigma\VV}{T_A}$.

\end{enumerate}
Observe that case~\ref{case:spine eval 1} implies case~\ref{case:spine eval 2} by expanding the memory $\e$ to $S_A$.

\begin{itemize}

\item\case{Unit case:}
We prove statement~\ref{case:spine eval 2}. Let $d=0$. We have $\term{\sigma(*)}=\term*$ and $\eval\e*\e$.

\item\case{Variable case:}
We prove statement~\ref{case:spine eval 2}. Let $d=0$. Given $\sigma$, since $\term{\sigma x}$ is a term of dimension $d$, we have $\eval\e{\sigma x}{T_A}$ for some memory $T_A$ of dimension $d$.

\item\case{Sequencing case:}
We prove statement~\ref{case:spine eval 2}. Given $\term{x;\WW}$, by induction we get a dimension $d$ such that statement~\ref{case:spine eval 1} holds for $\term\WW$. Let $\sigma$ be a substitution map of dimension $d$. As in the previous case, we have $\eval\e{\sigma x}{S_A}$ where $S_A$ is a memory of dimension $d$. Then by the inductive hypothesis, we have $\eval {S_A}{\sigma\WW}{T_A}$, and by the evaluation rule for sequencing, $\eval\e{\sigma(x;\WW)}{T_A}$.
\[
\evalRule
    {\eval\e{\sigma x}{S_A} \qquad \eval{S_A}{\sigma\WW}{T_A}}
	{\eval\e{\sigma(x;\WW)}{T_A}} 
\]

\item\case{Application case:}
We prove statement~\ref{case:spine eval 2}. Given $\term{[M]a.\VV}$, by induction we get a dimension $d$ such that $\eval\e{\sigma\VV}{T_A}$ given any $\sigma$ of dimension $d$. Then by expansion and the evaluation rule for application, statement~\ref{case:spine eval 2} follows for $\term{[M]a.\VV}$ with dimension $d$ and return memory $\term{\sigma M}_a\msep T_A$ (abusing notation to prepend $\term{\sigma M}$ to the stack $a$ in $T_A$).
\[
	\evalAlign
	  {\term{\sigma M}_a}       {\sigma\VV}  {\term{\sigma M}_a\msep T_A}
	  {\e}                {\sigma([M]a.\VV)} {\term{\sigma M}_a\msep T_A} 
\]

\item\case{Abstraction case:}
Given $\term{a<x>.\WW}$, by induction there is a dimension $d$ such that statement~\ref{case:spine eval 1} holds for $\term\WW$. We prove statement~\ref{case:spine eval 1} for dimension $d+1$. A memory of dimension $d+1$ is of the form $S_A\msep\term M_a$ where both $\term{M}$ and $S_A$ are of dimension $d$, in the latter case since every stack is of length at least $d$ (and holding terms of dimenstion at least $d$). A substitution map $\sigma$ of dimension $d+1$ is also of dimension $d$. Then so is $\sigma\{\term M/x\}$ (which assigns $\term M$ to $\term x$ and behaves as $\sigma$ everywhere else). The inductive hypothesis then gives $\eval{S_A}{\sigma\{M/x\}\WW}{T_A}$. By the evaluation rule for abstraction, from this we get $\eval{S_A\msep\term M_a}{\sigma(a<x>.\WW)}{T_A}$ for dimension $d+1$.
\[
	\evalAlign
	   {S_A}              {\sigma\{M/x\}\WW} {T_A}
	   {S_A\msep\term M_a} {\sigma(a<x>.\WW)}{T_A}
\qedhere
\]
\end{itemize}
\end{proof}

\spineTermination*

\begin{proof}
By induction on the spine context for the reduction step. The first six cases cover reduction at top level, and the last four reduction in context.

\begin{itemize}

\item\case{Beta case:} Let $\term{[N]a.a<x>.M}\rw\term{\{N/x\}M}$. Noting that $\term{\sigma(\{N/x\}M)}=\term{(\sigma\{\sigma N/x\})M}$, the left hand term evaluates as below if and only if the right hand term evaluates as $\eval{S_A}{\sigma(\{N/x\}M)}{T_A}$.
\[
\begin{array}{@{}r@{}l@{}}
	      \strut S_A                         ,\,    \term{\sigma(\{N/x\}M)} & {}\evalarrow T_A
\\ \hline \strut S_A \msep \term{\sigma N}_a ,\,      \term{\sigma(a<x>.M)} & {}\evalarrow T_A
\\ \hline \strut S_A                         ,\, \term{\sigma([N]a.a<x>.M)} & {}\evalarrow T_A
\end{array}
\]

\item\case{Passage case:} Let $\term{[N]b.a<x>.M}\rw\term{a<x>.[N]b.M}$ where $x$ is not free in $\term N$. Assuming by alpha-conversion that $\sigma$ does not bind in $\term P$, we have the first derivation if and only if we have the second.
\[
\begin{array}{@{}r@{}l@{}}
	      \strut S_A \msep \term{\sigma N}_b                  ,\,    \term{\sigma(\{P/x\}M)} & {}\evalarrow T_A
\\ \hline \strut S_A \msep \term{P}_a \msep \term{\sigma N}_b ,\,      \term{\sigma(a<x>.M)} & {}\evalarrow T_A
\\ \hline \strut S_A \msep \term{P}_a                         ,\, \term{\sigma([N]b.a<x>.M)} & {}\evalarrow T_A
\end{array}
\qquad
\begin{array}{@{}r@{}l@{}}
	      \strut S_A \msep \term{\sigma N}_b  ,\,      \term{\sigma(\{P/x\}M)} & {}\evalarrow T_A
\\ \hline \strut S_A                          ,\, \term{\sigma([N]b.\{P/x\}M)} & {}\evalarrow T_A
\\ \hline \strut S_A \msep \term{P}_a         ,\,   \term{\sigma(a<x>.[N]b.M)} & {}\evalarrow T_A
\end{array}
\]

\item\case{Next case:} Let $\term{*;M}\rw\term M$. Since $\term{\sigma(*)}=\term*$, we have $\eval{S_A}{\sigma M}{T_A}$ if and only if we have the folowing.
\[
\infer{ \eval{S_A}{\sigma(*;M)}{T_A} }{ \eval{S_A}*{S_A} && \eval{S_A}{\sigma M}{T_A} }
\]

\item\case{Prefix (pop) case:} Let $\term{a<x>.N;M}\rw\term{a<x>.(N;M)}$ where $\term x$ is not free in $\term M$. Assuming by alpha-conversion that $\sigma$ does not bind in $\term P$, we have the first derivation if and only if we have the second.
\[
\infer{ \eval{R_A\msep\term{P}_a}{\sigma(a<x>.N;M)}{T_A} }{
 \infer{ \eval{R_A\msep\term{P}_a}{\sigma(a<x>.N)}{S_A} }{
  \eval{R_A}{\sigma(\{P/x\}N)}{S_A}
 } &&
 \eval{S_A}{\sigma M}{T_A}
}
\qquad
\infer{ \eval{R_A\msep\term{P}_a}{\sigma(a<x>.(N;M))}{T_A} }{
 \infer{ \eval{R_A}{\sigma(\{P/x\}N;M)}{T_A} }{
  \eval{R_A}{\sigma(\{P/x\}N)}{S_A}
  &&
  \eval{S_A}{\sigma M}{T_A}
}}
\]

\item\case{Prefix (push) case:} For a reduction $\term{[P].N;M}\rw\term{[P].(N;M)}$, we have the first derivation if and only if we have the second.
\[
\infer{ \eval{R_A}{\sigma([P].N;M)}{T_A} }{
 \infer{ \eval{R_A}{\sigma([P].N)}{S_A} }{
  \eval{R_A\msep\term{\sigma P}_a}{\sigma N}{S_A}
 } &&
 \eval {S_A}{\sigma M}{T_A}
}
\qquad
\infer{ \eval{R_A}{\sigma([P].(N;M))}{T_A} }{
 \infer{ \eval{R_A\msep\term{\sigma P}_a}{\sigma(N;M)}{T_A} }{
  \eval{R_A\msep\term{\sigma P}_a}{\sigma N}{S_A}
  &&
  \eval {S_A}{\sigma M}{T_A}
}}
\]

\item\case{Associativity case:} For a reduction $\term{(P;N);M}\rw\term{P;(N;M)}$, we have the first derivation if and only if we have the second.
\[
\infer{ \eval{R_A}{\sigma((P;N);M)}{U_A} }{
 \infer{ \eval{R_A}{\sigma(P;N)}{T_A} }{
  \eval {R_A}{\sigma P}{S_A}
  &&
  \eval {S_A}{\sigma N}{T_A}
 } &&
 \eval {T_A}{\sigma M}{U_A}
}
\]
\[
\infer{ \eval{R_A}{\sigma(P;(N;M))}{U_A} }{
 \eval {R_A}{\sigma P}{S_A}
 &&
 \infer{ \eval{S_A}{\sigma(N;M)}{U_A} }{
  \eval {S_A}{\sigma N}{T_A}
  &&
  \eval {T_A}{\sigma M}{U_A}
}}
\]

\item\case{Application case:} Let $\term{[P]a.M}\rw\term{[P]a.N}$ where $\term M\rw\term N$. Given $S_A$ and $\sigma$, by induction we have $\eval{S_A\msep\term{\sigma P}_a}{\sigma M}{T_A}$ if and only if $\eval{S_A\msep\term{\sigma P}_a}{\sigma N}{T_A}$, and hence $\eval{S_A}{\sigma([P]a.M)}{T_A}$ if and only if $\eval{S_A}{\sigma([P]a.N)}{T_A}$.

\item\case{Abstraction case:} Let $\term{a<x>.M}\rw\term{a<x>.N}$ where $\term M\rw\term N$. Given $S_A\msep\term{\sigma P}_a$ and $\sigma$, assuming $\sigma$ does not bind in $\term P$, by induction we have $\eval{S_A}{\sigma(\{P/x\}M)}{T_A}$ if and only if $\eval{S_A}{\sigma(\{P/x\}N)}{T_A}$, and hence $\eval{S_A\msep\term{\sigma P}_a}{\sigma(a<x>.M)}{T_A}$ if and only if $\eval{S_A\msep\term{\sigma P}_a}{\sigma(a<x>.N)}{T_A}$.

\item\case{Prefix case:} Let $\term{M;P}\rw\term{N;P}$ where $\term M\rw\term N$. Given $S$, by induction we have $\eval{R_A}{\sigma M}{S_A}$ if and only if $\eval{R_A}{\sigma N}{S_A}$, so given $\eval{S_A}{\sigma P}{T_A}$ we have $\eval{R_A}{\sigma(M;P)}{T_A}$ if and only if $\eval{S_A}{\sigma(N;P)}{T_A}$.

\item\case{Postfix case:} Let $\term{P;M}\rw\term{P;N}$ where $\term M\rw\term N$. By induction we have $\eval{S_A}{\sigma M}{T_A}$ if and only if $\eval{S_A}{\sigma N}{T_A}$, so given $\eval{R_A}{\sigma P}{S_A}$ we have $\eval{R_A}{\sigma(P;M)}{T_A}$ if and only if $\eval{R_A}{\sigma(P;N)}{T_A}$. 
\qedhere
\end{itemize}
\end{proof}

\section{Proofs for section~\ref{sec:SN}: Strong Normalization}

\substitutionStrong*

\begin{proof}
The proof is similar to the substitution lemma for the weak type system, Lemma~\ref{lem:substitutionWeak}. We cover the new cases: \emph{strong application} and \emph{weakening}.

\begin{itemize}

	\item\case{Strong application case:} 
For a term $\term M=\term{[Q]a.P}$ we have the following derivations, both of the same weight, where $\type{n = n_1+n_2+n_3}$,  $\term{G = G_1+G_2+G_3}$ and $\term{D = D_1+D_2+D_3}$. In the second derivation, the premises of the top rules are stacked for space and readability. The argument is as follows. Given the first derivation, where $\term{G |- N:n}$ is introduced by a collection typing rule $\RR c$, we may partition its premisses along the decomposition $\type{n = n_1+n_2+n_3}$ to get $\term{G_i |- N:n_i}$ for $i=1,2,3$. Then we may construct the second derivation. In the reverse direction, given the second derivation, we may construct the first.
\[
\infer-[\rr s]{\term{G+D |- \{N/x\}[Q]a.P: !!k => !!m}}{
   \term{G |- N:n}
&& \infer[\rr A]{\term{x:n , D |s [Q]a.P : !!k => !!m}}{
     \term{x:n_1, D_1  |- Q : i}
   & \term{x:n_2, D_2  |- P : a(i)\,!!k => !!m}
   & \term{x:n_3, D_3  |- Q : t}
   }
}
\]
\[
\infer[\rr A]{\term{G+D |s [\{N/x\}Q].\{N/x\}P: !k => !m}}{
   \infer-[\rr s]{\term{G_1+D_1 |- \{N/x\}Q : i}}{
   \begin{array}{@{}r@{}l@{}}  
   	          \term{G_1} & \term{|- N:n_1}
   \\ \term{x:n_1 , D_1} & \term{|- Q:i}
   \end{array}
   }
&  \infer-[\rr s]{\term{G_2+D_2 |- \{N/x\}P : a(i)\,!!k => !!m}}{
   \begin{array}{@{}r@{}l@{}}  
   	          \term{G_2} & \term{|- N:n_2}
   \\ \term{x:n_2 , D_2} & \term{|- P : a(i)\,!!k => !!m}
   \end{array}
   }
&  \infer-[\rr s]{\term{G_3+D_3 |- \{N/x\}Q : t}}{
   \begin{array}{@{}r@{}l@{}}  
   	          \term{G_3} & \term{|- N:n_3}
   \\ \term{x:n_3 , D_3} & \term{|- Q:t}
   \end{array}
   }
}
\]

	\item\case{Weakening case:} 
From left to right, assume the derivation below left with sub-derivations of weights $i$ and $k$. The derivation for $\term N$, ending in a collection typing rule, may be partitioned along $\type{n+m}$ to give $\term{G_1 |_n N:n}$ and $\term{G_2 |_m N:m}$ where $n+m=i$. Then we may construct the derivation below right, which by induction has a weight $j\leq n+k$.
\[
\infer-[\rr s]{\term{G+D+E |- \{N/x\}M : t}}{
   \term{G |_i N:n+m}
&& \infer[\rr W]{\term{x:n+m , D + E |s M : t}}{
     \term{x:n , D |_k M : t}
}}
\qquad
\infer[\rr W]{\term{G_1+G_2+D+E |- \{N/x\}M : t}}{
  \infer-[\rr s]{\term{G_1+D+E |- \{N/x\}M : t}}{
     \term{G_1 |_n N:n}
  && \term{x:n , D |_k M : t}    
}}
\]
From right to left, given the derivation below left we may construct that below right of the same weight $n+m$, where by alpha-conversion we may assume $x$ does not occur in the context $\term{E}$.
\[
\infer[\rr W]{\term{G+D+E |- \{N/x\}M : t}}{
  \infer-[\rr s]{\term{G+D |- \{N/x\}M : t}}{
     \term{G |_n N:n}
  && \term{x:n , D |_m M : t}
}}
\qquad
\infer-[\rr s]{\term{G+D+E |- \{N/x\}M : t}}{
   \term{G |_n N:n}
&& \infer[\rr W]{\term{x:n , D+E |- M : t}}{
     \term{x:n , D |_m M : t}
}}
\qedhere
\]
\end{itemize}
\end{proof}

\subjectReductionBeta*

\begin{proof}
By induction on the reduction context. We give the case for a beta-redex and that for the \emph{argument} context $\term{[C]a.M}$ of the strong application rule. For the other cases, the derivation for $\term{C\{M\}}$ contains exactly one subderivation for $\term{M}$, so that a change in weight is reflected directly.

\begin{itemize}

	\item\case{Beta case:} 
For a reduction $\term{[N]a.a<x>.M}\rw\term{\{N/x\}M}$, given the first derivation below we may construct the second. Where the first derivation has weight $n+m+k+2$, by Lemma~\ref{lem:strong substitution} the second has weight $i\leq n+m$.
\[
\infer[\rr A]{\term{G+D+E |s [N]a.a<x>.M : !!k => !!m}}{
   \term{G |_n N:i}
&& \infer[\rr l]{\term{D |- a<x>.M: a(i)\,!!k => !!m}}{
      \term{D , x:i |_m M: !!k => !!m}
   }
&& \term{E |_k N:t}
}
\]
\[
\infer[\rr W]{\term{G+D+E |s \{N/x\}M : !!k => !!m}}{
   \infer-[\rr s]{\term{G+D |- \{N/x\}M : !!k => !!m}}{
      \term{G |_n N:i}
   && \term{D , x:i |_m M: !!k => !!m}
   }
}
\]

	\item\case{Argument case:} 
For a reduction $\term{[M]a.P}\rw\term{[N]a.P}$ given $\term M\rw\term N$, consider the strong application inference below, where $\term{G=G_1++G_n}$ and $\type{i=[t_1..t_n]}$ Observe that there is at least one sub-derivation for $\term M$, given by $\term{E |- M:t}$, in the derivation for $\term{[M]a.P}$. Replacing every sub-derivation for $\term M$ with one for $\term N$ of strictly smaller weight then gives a strictly smaller weight for the derivation for $\term{[N]a.P}$.
\[
\infer[\rr A]{\term{G+D |s [M]a.P : !!k => !!m}}{
   \infer[\rr c]{\term{G |- M: i}}{(\term{G_i|- M:t_i})_{1\leq i\leq n}}
&& \term{D |- P: a(i)\,!!k => !!m}
&& \term{E |- M: t}
}
\qedhere
\]
\end{itemize}
\end{proof}

\strongNonBeta*

\begin{proof}
The proof is similar to Lemma~\ref{lem:spine reduction}. The one new case for top-level reduction is the following.
\begin{itemize}

	\item\case{Strong application/sequence case:} 
For a reduction $\term{[P]a.N;M}\rw\term{[P]a.(N;M)}$, given the first derivation below we may construct the second, of equal weight.
\[
\infer[\rr ;]{\term{G+D+E+F |- [P]a.N;M : !!k => !!n }}{
   \infer[\rr A]{\term{G+D+E |s [P]a.N : !!k => !!m}}{
      \term{G |- P: i}
   && \term{D |- N: a(i)\,!!k => !!m}
   && \term{E |- P: t}
   }
&& \term{F |- M : !!m => !!n}   
}
\]
\[
\infer[\rr A]{\term{G+D+E+F |s [P]a.(N;M) : !!k => !!n }}{
   \term{G |- P: i}
&& \infer[\rr ;]{\term{D+F |- N;M : a(i)\,!!k => !!n}}{
      \term{D |- N: a(i)\,!!k => !!m}
   && \term{F |- M : !!m => !!n}  
   } 
&& \term{E |- P: t}
}
\]
\end{itemize}
For reduction in context, $\term{C\{M\}}\rw\term{C\{N\}}$, the new case is the argument position, where there may be any (non-zero) number of sub-derivations for $\term M$ in $\term{C\{M\}}$. Replacing a sub-derivation with one of equal or lower weight then results in an overall weight that is equal or lower.
\end{proof}

\section{Proofs for Section~\ref{sec:perp}: Perpetual Evaluation}

\normalforms*

\begin{proof}
The proof is analogous to that of Proposition~\ref{prop:spine NF}. Terms $\term{\MM}$ are the normal forms, and $\term{\NN}$ those that are not abstractions. For the latter, variables $\term x$ and the unit $\term *$ are normal, and a composition $\term{M;N}$ is normal if $\term M$ is a variable and $\term N$ is normal, and thus of the form $\term{x;\MM}$. An application $\term{[N]a.M}$ is normal if $\term N$ and $\term M$ are and $\term M$ is not an abstraction, thus of the form $\term{[\MM].\NN}$. Finally, an abstraction $\term{a<x>.M}$ is normal if $\term M$ is, thus of the form $\term{a<x>.\MM}$.
\end{proof}

%
%
%
%

\SNperpetual*

\begin{proof}
The proof is by induction on the following pair of measures: first, the sum length of reduction paths from $\term M$, and second, the size of $\term M$. Let $\term M$ be of the form $\term{[S_A].M'}$, with $\term{M'}$ not an application and proceed by case analysis on the structure of $\term{M'}$.
\begin{itemize}
	
	\item\case{Unit case:} if $\term M$ is of the form $\term{[S_A].*}$ then we have the case \emph{Normal (unit)} of the definition of $\rp$ as below. The premisses are given by the inductive hypothesis: each term $\term{N_i}$ has equal or fewer reductions than $\term M$, and is smaller. The reduction $\term M\rwn\term\MM$ follows from the reductions $\term{N_i}\rwn\term{\MM_i}$.
\[
\infer{\term{[N_1]a_1..[N_n]a_n.*} \rp \term{[\MM_1]a_1..[\MM_n]a_n.*}}{(\term{N_i}\rp\term{\MM_i})_{i\leq n}}
\]

	\item\case{Variable case:} if $\term M$ is of the form $\term{[S_A].x}$ then we have the case \emph{Normal (variable)} as below. It resolves as the previous case.
\[
\infer{\term{[N_1]a_1..[N_n]a_n.x} \rp \term{[\MM_1]a_1..[\MM_n]a_n.x}}{(\term{N_i}\rp\term{\MM_i})_{i\leq n}}
\]

	\item\case{Abstraction case:} if $\term{M'}$ an abstraction $\term{a<x>.N}$ there are three sub-cases. If $\term{[S_A]}$ is empty, we have the case \emph{Normal (abstraction)} as below.
\[
\infer{\term{a<x>.N}\rp\term{a<x>.\MM}}{\term{N}\rp\term \MM}
\]
If $\term{[S_A]}$ is of the form $\term{[T_A].[P]a}$, ending in a push on location $a$, we have the case \emph{Beta} as below, with the reduction $\term{[T_A].[P]a.a<x>.N}\rw\term{[T_A].\{P/x\}N}\rw\term\MM$.
\[
\infer{\term{[T_A].[P]a.a<x>.N}\rp\term\MM}{\term{[T_A].\{P/x\}N}\rp\term\MM}
\]
If $\term{[S_A]}$ is of the form $\term{[T_A].[P]b}$, ending in a push on location $b\neq a$, we have the case \emph{Passage} as below, with the reduction $\term{[T_A].[P]b.a<x>.N}\rw\term{[T_A].a<x>.[P]b.N}\rw\term\MM$.
\[
\infer{\term{[T_A].[P]b.a<x>.N}\rp\term\MM}{\term{[T_A].a<x>.[P]b.N}\rp\term\MM}
\]
	\item\case{Unit/sequence case:}
If $\term M$ is of the form $\term{[S_A].(*;P)}$ we have the case \emph{Next} below, with the reduction $\term{[S_A].(*;P)}\rw\term{[S_A].P}\rwn\term\MM$. 
\[
\infer{\term{[S_A].(*;P)} \rp \term\MM}{\term{[S_A].P}\rp \term\MM}
\]

	\item\case{Variable/sequence case:}
If $\term M$ is of the form $\term{[S_A].(x;P)}$ we have the case \emph{Normal (sequence)} as below. The inductive hypothesis gives the premisses, which each have a smaller measure. The reduction from $\term M$ follows by that for each $\term{N_i}$ and $\term P$.
\[
\infer{\term{[N_1]a_1..[N_n]a_n.(x;P)} \rp \term{[\MM_1]a_1..[\MM_n]a_n.(x;\MM)}}{(\term{N_i}\rp\term{\MM_i})_{i\leq n} && \term P\rp\term\MM}
\]

	\item\case{Abstraction/sequence case:}
If $\term M$ is of the form $\term{[S_A].(a<x>.N;P)}$ we have the case \emph{Prefix (pop)} as below, with the reduction $\term{[S_A].(a<x>.N ; P)} \rw \term{[S_A].a<x>.(N ; P)}\rwn\term\MM$.
\[
\infer{\term{[S_A].(a<x>.N ; P)} \rp \term\MM}{\term{[S_A].a<x>.(N ; P)}\rp \term\MM}
\]

	\item\case{Application/sequence case:}
If $\term M$ is of the form $\term{[S_A].([P]a.N;Q)}$ we have the case \emph{Prefix (push)} as below, with the reduction $\term{[S_A].([P]a.N ; Q)} \rw \term{[S_A].[P]a.(N ; Q)}\rwn \term\MM$.
\[
\infer{\term{[S_A].([P]a.N ; Q)} \rp \term\MM}{\term{[S_A].[P]a.(N ; Q)}\rp \term\MM}
\]

	\item\case{Sequence/sequence case:}
If $\term M$ is of the form $\term{[S_A].((Q;P);N)}$ we have the case \emph{Associate} as below, with the reduction $\term{[S_A].((Q;P);N)}\rw\term{[S_A].(Q;(P;N))}\rwn\term\MM$.
\[
\infer{\term{[S_A].((Q;P);N)} \rp \term\MM}{\term{[S_A].(Q;(P;N))}\rp \term\MM}  \qedhere
\]
\end{itemize}
\end{proof}

\perpetualTyped*

\begin{proof}
For the notation $\term{[S_A].M}$ we introduce an associated typing rule~$\RR S$ (\emph{applications}), and a strong typing rule for memories $\RR O$.
\[
	\infer-[\rr S]{\term{G+D |s [S_A].M : !!k => !!m}}{\term{G |s ?{S_A} : !!i} && \term{D |s M : !!i\,!!k => !!m}}
\qquad
	\infer-[\rr O]{\term{G+D+E |s ?{S_A\msep\term N}_a : !!i\,a(k)}}{\term{G |s ?{S_A}:!!i} && \term{D |s N: k} && \term{E |s N:t}}
\]
By induction, these rule are admissible: we have the first derivation below if and only if we have the second.
\[
\infer-[\rr S]{\term{G+D+E+F |s [S_A\msep N_a].M}}{
  \infer-[\rr O]{\term{G+D+E |s ?{S_A\msep\term N_a} : !!i\,a(k)}}{
	\term{G |s S_A : !!i}
  & \term{D |s N : k}
  & \term{E |s N : t}
  }
& \term{F |s M : a(k)\,!!i\,!!m => !!n}
}
\]
\[
\infer-[\rr S]{\term{G+D+E+F |s [S_A].[N]a.M : !!m => !!n}}{ 
  \term{G |s S_A : !!i}
& \infer[\rr A]{\term{D+E+F |s [N]a.M : !!i\,!!m => !!n}}{
    \term{D |s N : k}
  & \term{E |s N : t}
  & \term{F |s M : a(k)\,!!i\,!!m => !!n}
} }
\]
Again by induction, 
a memory $S_A=\e\msep(\term{N_1})_{a_1}\cdots(\term{N_n})_{a_n}$ may be typed $\type{[]^f}$ where $f=|S_A|$, if there is a type for each $\term{N_i}$. This gives the following admissible rules.
\[
\infer-{\term{G_1++G_n |s ?{S_A} : []^f }}{(\term{G_i |- N_i : t_i})_{i\leq n}}
\qquad
\infer-{\term{G_1++G_n + D |s [S_A].M : !!i => !!k }}{(\term{G_i |- N_i : t_i})_{i\leq n} && \term{D |- M : []^f\,!!i => !!k}}
\]
The proof is then by induction on the relation $\rp$.
\begin{itemize}

	\item\case{Beta case:}
\[
\infer
    {\term{[S_A].[N]a.a<x>.P} \rp \term\MM}
    {\term{[S_A].\{N/x\}P} \rp \term\MM && \term N\rp\term{\MM'}}
\]
The inductive hypothesis gives the following three derivations: that below (using the strong substitution lemma, Lemma~\ref{lem:strong substitution}), $\term{F |- N : t}$, and $\term{G+D+E |- \MM : !!m => !!n}$.
\[
\infer-[\rr S]{\term{G+D+E |- [S_A].\{N/x\}P : !!m => !!n}}{
  \term{G |- ?{S_A} : !!i}
& \infer-[\rr s]{\term{D+E |- \{N/x\}P : !!i\,!!m => !!n}}{
    \term{D |- N : k}
  & \term{E , x: k |- M : !!i\,!!m => !!n}
} }
\]
From these we may construct the following derivations, the first for $\term M$, the second for $\term\MM$.
\[
\infer-[\rr S]{\term{G+D+E+F |- [S_A].[N]a.a<x>.P}}{
  \term{G |- ?{S_A} : !!i}
& \infer[\rr A]{\term{D+E+F |- [N]a.a<x>.P : !!i\,!!m => !!n}}{
    \term{D |- N : k}
  & \infer[\rr l]{\term{E |- a<x>.P : a(k)\,!!i\,!!m => !!n}}{
      \term{E , x: k |- M : !!i\,!!m => !!n}
    }
  & \term{F |- N : t}
} }
\]
\[
\infer[\rr W]{\term{G+D+E+F |- \MM : !!m => !!n}}{\term{G+D+E |- \MM : !!m => !!n}}
\]

	\item\case{Passage, Next, Prefix (pop), Prefix (push), and Associate cases:}
All five cases are of the following form, where $\term M\rw\term N$ by a non-beta reduction.
\[
\infer{\term M\rp\term \MM}{\term N\rp\term\MM}
\]
The inductive hypothesis gives derivations $\term{G |s N : t}$ and $\term{G |s \MM : t}$. The required derivation $\term{G |s M : t}$ follows by Lemma~\ref{lem:non-beta subject reduction} (non-beta reduction preserves and reflects typing).

	\item\case{Normal (abstraction) case:}
\[
\infer{\term{a<x>.M}\rp\term{a<x>.\MM}}{\term M\rp\term \MM}
\]
The inductive hypothesis gives $\term{G , x:i |- M : !!k => !!m}$ and $\term{G , x:i |- \MM : !!k => !!m}$ (where $\type i=\type{[]}$ if $x$ does not occur). The abstraction typing rule $\RR l$ gives $\term{G |- a<x>.M : a(i)\,!!k => !!m}$ and $\term{G |- a<x>.\MM : a(i)\,!!k => !!m}$.

	\item\case{Normal (unit) case:}
\[
\infer{\term{[N_1]a_1..[N_n]a_n.*} \rp \term{[\MM_1]a_1..[\MM_n]a_n.*}}{(\term{N_i}\rp\term{\MM_i})_{i\leq n}}
\]
The inductive hypothesis gives $\term{G_i |- N_i : t_i}$ and $\term{G_i |- \MM_i : t_i}$ for each $i\leq n$. By the admissible rules introduced above, we have the following derivation for $\term M$, and similarly for $\term\MM$, where $f$ is the dimension of the memory $\e\msep(\term{N_1})_{a_1}\cdots\msep(\term{N_n})_{a_n}$.
\[
\infer-{\term{G_1++G_n |- [N_1]a_1..[N_n]a_n.* : \e => []^f}}{(\term{G_i |- N_i : t_i})_{i\leq n} && \term{|- * : []^f => []^f}}
\]

	\item\case{Normal (variable) case:}
\[
\infer{\term{[N_1]a_1..[N_n]a_n.x} \rp \term{[\MM_1]a_1..[\MM_n]a_n.x}}{(\term{N_i}\rp\term{\MM_i})_{i\leq n}}
\]
The case is as the previous one, with the following derivation for $\term M$.
\[
\infer-{\term{G_1++G_n + x:[\,[]^f => \e\,] |- [N_1]a_1..[N_n]a_n.x : \e => \e}}{(\term{G_i |- N_i : t_i})_{i\leq n} && \term{x:[\,[]^f => \e\,] |- x : []^f => \e}}
\]

	\item\case{Normal (sequence) case:}
\[
\infer
    {\term{[N_1]a_1..[N_n]a_n.(x;M)} \rp \term{[\MM_1]a_1..[\MM_n]a_n.(x;\MM)}}
    {(\term{N_i}\rp\term{\MM_i})_{i\leq n} && \term M\rp\term\MM}
\]
The inductive hypothesis gives $\term{D |- M: !!k => !!m}$ and $\term{D |- \MM: !!k => !!m}$. The case is then as the previous two, with the following derivation.
\[
\infer-{\term{G_1++G_n + x:[\,[]^f => !!k\,] + D |- [N_1]a_1..[N_n]a_n.(x;M) : \e => !!m }}{
  (\term{G_i |- N_i : t_i})_{i\leq n}
& \infer[\rr ;]{\term{x:[\,[]^f => !!k\,] + D |- x;M : []^f => !!m }}{
    \term{x:[\,[]^f => !!k\,] |- x: []^f => !!k} 
  & \term{D |- M: !!k => !!m}
} }
\qedhere
\]
\end{itemize}
\end{proof}


\end{document}